\documentclass[12pt,twoside]{amsart}
\usepackage{amsfonts}
\usepackage{amssymb}
\usepackage{amsthm}
\usepackage{amsmath}
\usepackage{newlfont}
\usepackage{t1enc}
\usepackage{bbm}
\usepackage{fontsmpl}
\usepackage{graphics}
\usepackage{graphicx}
\usepackage{color}
\usepackage{psfrag}
\usepackage{a4wide}
\usepackage{amssymb}
\usepackage{enumerate}
\usepackage{hyperref}

\theoremstyle{remark}

\theoremstyle{definition}
\newtheorem{teo}{Theorem}[section]

\newtheorem{ejem}[teo]{Examples}
\newtheorem{lema}[teo]{Lemma}
\newtheorem{prop}[teo]{Proposition}
\newtheorem{cor}[teo]{Corollary}
\newtheorem{obs}[teo]{Remark}
\newtheorem{defi}[teo]{Definition}

\newtheorem{notat}[teo]{Notation}

\newcommand{\aref}[1]{(\ref{#1})}

\newcommand{\R}{{\mathbb R}}
\newcommand{\C}{{\mathcal C}}
\newcommand{\Z}{{\mathbb{Z}}}

\newcommand{\N}{{\mathbb N}}

\newcommand{\F}{{\mathcal F}}

\newcommand{\B}{{\mathcal B}}

\def\A{\mathcal A}

\def\F{\mathcal F}

\def\1{{\mathbf 1}}

\def\R{\mathbb{R}}
\def\N{\mathbb{N}}

\def\A{\mathcal A}

\def\F{\mathcal F}

\newcommand{\rojo}[1]{\textcolor{red}{#1}}

\numberwithin{equation}{section}

\begin{document}

\title{Stability of gas measures under perturbations and discretizations}
\author{Roberto Fern\'andez, Pablo Groisman and Santiago Saglietti}
\date{}
\begin{abstract} For a general class of gas models ---which includes discrete and continuous Gibbsian models as well as contour or polymer ensembles--- we determine a \emph{diluteness condition} that implies: (1) Uniqueness of the infinite-volume equilibrium measure; (2) stability of this measure under perturbations of parameters and discretization schemes, and (3)  existence of a coupled perfect-simulation scheme for the infinite-volume measure together with its perturbations and discretizations. Some of these results have previously been obtained through methods based on cluster expansions.  In contrast, our treatment is purely probabilistic and its diluteness condition is weaker than existing convergence conditions for cluster expansions.
\end{abstract}
\maketitle

\section{Introduction}

Phase transitions in statistical mechanics are often studied through sequences of models involving convergent sequences of parameters. The analysis is usually ``towards a target'': properties of the target model are inferred from properties of models in the sequence. For instance, sequences of models with asymptotically vanishing fields can be used to show first order phase transitions for the zero-field model. In the same spirit, transitions in continuum-space models are often studied through limits of models on lattices with decreasing mesh.  

In this paper we treat the opposite ---``from the target''--- point of view. We investigate conditions under which properties of a target model are inherited by models  obtained by perturbing parameters (e.g. fields or fugacities) or the configuration space itself (e.g. through discretizations). While our treatment is general, our basic motivation came from target models in the continuum, for which we wished to address two types of issues:   

\begin{description}
	\item [Faithfulness of simulation schemes] Simulation of continuum models requires unavoidable discretizations. It is tacitly understood that using sufficiently refined discretizations leads to trustable determinations of phase diagrams. Still, one may wonder if this is always the case. The question involves, in fact, a previous issue:  How should the discretization be performed? The natural choice would be to discretize the model (both configuration space and interactions) and sample from the Gibbs measure in the discrete system. An alternative, however, would be to sample the discretized version of the actual continuum measure. These approaches are quite likely not equivalent in general, as it is known that coarse-graining may lead to non-Gibbsianness.
	
	\item[Universality] Stability of the continuum equilibrium measure under discretizations should imply the irrelevance in the discretized systems of  interaction terms which disappear in the continuum limit. This includes, for instance, \mbox{hard-core exclusions} involving events with zero probability in the continuum (see the thin rods model in Section \ref{sec:dtrm} below, for example). In this sense, stability can be \mbox{interpreted as} the continuum model acting as a ``universality class'' for whole families of discrete systems.
\end{description}

Our treatment is geared towards general \emph{gas models}, that is systems involving families of geometrical objects ---possibly with further decorations such as color or spin--- distributed on some underlying space. This distribution is assumed to be, while on bounded volumes, absolutely continuous with respect to a basic ``free'' Poisson process.
Our approach uses only general properties of this density with respect to the ``free measure'', and hence it is applicable to general point processes not necessarily endowed with a Gibbsian description. In particular, it applies to the contour ensembles used to describe low-temperature phases, starting with the well-known Peierls contours.   

The results reported below hold for models satisfying an appropriate \emph{dilution condition} which physically corresponds to gaseous phases. This dilution condition leads to objects typically clustering into finite islands separated by percolating empty space. In particular, our condition implies uniqueness of the infinite-volume point process or Gibbs measure. In the Gibbsian setting, our requirement leads naturally to models at high temperature or low fugacity. However, the generality of our approach makes the technique relevant also for low-temperature (or condensation) regimes in which typical configurations can be described by diluted contours. This generality will be exploited in a forthcoming paper; in the present paper we focus on stability issues within the uniqueness regime.   

Our main result establishes that whenever a target model satisfies this dilution condition 
all sufficiently small perturbations of it also admit exactly one consistent infinite-volume measure which, furthermore, is a slight perturbation of the measure in the target model. Let us list some particularities of our approach:

\begin{itemize}
	\item The dilution condition amounts to a strong form of uniqueness that is, however, weaker than the condition associated to the validity of cluster expansion methods and therefore applies to a wider range of systems. This extension comes at a cost: Roughly speaking, our condition implies only the continuity of the \mbox{unique measure} with respect to slight perturbations in the parameters of the model, whereas \mbox{the convergence} of cluster expansions leads to analytic dependences.  
	
	\item Our approach yields a coupled perfect simulation algorithm, that is an algorithm yielding simultaneously exact samples of the infinite-volume equilibrium measures (restricted to a finite volume) of target and approximating models. 
	
	\item This coupled algorithm leads to the almost sure convergence of the samples above, and thus to the weak convergence of the equilibrium measures of perturbed models towards that of the target model.
	
\end{itemize}

We believe that our results are quite natural and easy to apply, as we illustrate through a number of examples. Nevertheless, our presentation is not devoid of technical details.  In particular, in Section \ref{sec:bs} we present a careful account of the general setup for gas models (which is used indistinctly for discrete and continuum systems) followed by the precise definition of ``approximation'' operations (i.e. perturbations of the configuration space). Also, in Section \ref{secffg} we discuss with some detail the \emph{ancestor algorithm} which constitutes the main tool of our analysis. This algorithm ---introduced 15 years ago as a substitute for cluster expansions in \cite{FFG1}--- reconstructs configurations through a (time-backwards) oriented percolation model of space-time ``cylinders'', i.e. objects in the gas model which live for a certain period of time. This ancestor algorithm succeeds ---implying uniqueness and space-time mixing of the infinite-volume measure--- if these cylinders do not percolate. Such a condition is naturally suited for stability studies, because finite cylinder clusters are robust under perturbations. Our dilution condition is crafted to ensure, in general terms, the lack of percolation for cylinders associated to the target model. This lack of percolation is then inherited by the perturbed models, whose cylinder clusters become in one-to-one correspondence with those of the target model as the strength of the perturbation vanishes.  

\section{The basic setup}\label{sec:bs}

\subsection{Configuration space}
We start by describing the general measure-theoretical setup. The definitions aim at a general configuration space on an underlying space of locations. Whenever the latter is discrete, we call the model in question a \textit{lattice system}, although other type of systems also fit into this framework.

\subsubsection{Particle configurations and configuration space}\label{sec:pc}

We consider two locally compact complete separable metric spaces: a \textit{location space} $(S,d_S)$ and a space of \textit{animals} $(G,d_G)$. A countable $S$ is often called a \textit{lattice} and a finite $G$ is interpreted as a set of \mbox{\textit{colors} or \textit{spins}.} The product space $S \times G$ is also locally compact complete and separable if endowed with the metric $d=d_S + d_G$. For convenience, we shall call an element $(x,\gamma) \in S \times G$ a \textit{particle} and denote it simply by $\gamma_x$. We interpret it as an animal $\gamma$ positioned at location $x$.

The general definition of configuration space requires special care. For our purposes, it will be convenient to adopt the general framework of point processes featured in \cite{K}. In this framework, configurations are identified with locally finite measures on $S \times G$ obtained as a superposition of delta-measures signaling the presence of particles. 

\begin{notat} Given a metric space $(X,d)$ we denote:
	\begin{itemize}
		\item [$\bullet$] By $\B_X$ the class of all Borel subsets of $(X,d)$.
		\item [$\bullet$] By $\B^0_X$ the set of elements of $\B_{X}$ with compact closure.
	\end{itemize}
\end{notat}
We recall that a set $B \in \B_X$ is \textit{locally finite} if for every $B' \in \B^0_X$ the set $B \cap B'$ is finite. Also, a measure $\xi$ on $(X, \B_{X})$ is called a \textit{Radon measure} if $\xi(B) < +\infty$ for every $B \in \B^0_{X}$. Configurations correspond to particular Radon measures supported on locally finite sets.

\begin{defi} \label{def:r1} A Radon measure $\xi$ on $(S\times G, \B_{S \times G})$ is said to be a \textit{particle configuration} if  $\xi(B) \in \N_0$ for every $B \in \B^0_{S \times G}$.
\end{defi}

The following proposition states that a particle configurations can be actually identified with a locally finite collection of particles, in which particles may appear more than once. 

\begin{prop}[{\cite[Lemma~2.1]{K}}] A measure $\xi$ on $S \times G$ is a particle configuration if and only if there exist a locally finite set $\langle \xi \rangle \subseteq S \times G$ and a map $m_\xi: \langle \xi \rangle \to \N$ such that
	\begin{equation}\label{standardrep}
	\xi = \sum_{\gamma_x \in \langle \xi \rangle} m_\xi(\gamma_x) \delta_{\gamma_x}
	\end{equation} with $\delta_{\gamma_x}$ the Dirac measure centered at $\gamma_x$. We call \eqref{standardrep} the \textit{standard representation} of $\xi$.
\end{prop}

Thus, a particle configuration $\xi$ can be thought of as a locally finite family of particles, each with its own respective multiplicity. The set $\langle \xi \rangle$ shall be called the \textit{support} of $\xi$. \mbox{It is the collection} of particles which constitute $\xi$, with no regard for their multiplicities. If we take these into consideration then we obtain the \textit{weighted support} of $\xi$ given by 
$$
[\xi]:=\bigl\{ (\gamma_x,i) \in (S\times G)\times \N : \gamma_x \in \langle \xi \rangle \text{ and }1 \leq i \leq m_\xi(\gamma_x)\bigr\}.
$$ 

Definition \ref{def:r1} alone is sufficient to define the configuration space whenever $G$ is compact. For the general case, however, we will require an extra restriction.

\begin{defi} \ 
	
	\begin{itemize}
		\item[i.] A measure $\xi$ on $(S\times G, \B_{S \times G})$ is said to be of \textit{$S$-locally finite allocation} if it satisfies $\xi(\Lambda \times G) < +\infty$ for every $\Lambda \in \B^0_S$.
		\item[ii.] The \textit{configuration space} of $S\times G$ is the space $\mathcal{N}(S\times G)$ of all particle configurations on $S \times G$ which are of $S$-locally finite allocation.
		\item [iii.] The set of configurations (supported) on $\Lambda \in \B^0_S$ is the set 
		$$
		\mathcal{N}(\Lambda \times G) \;:=\; \bigl\{ \xi\in \mathcal{N}(S \times G) : \langle\xi\rangle \subset \Lambda \times G\bigr\}.
		$$
		[The notation is slightly abusive.]
	\end{itemize}
\end{defi}

We now define the natural notions of restriction and superposition of configurations.


\begin{defi} Given $\xi \in \mathcal{N}(S \times G)$ and $A \in \B_{S \times G}$, the \textit{restriction of} $\xi$ \textit{to} $A$ is defined as the particle configuration $\xi_{A}$ such that for every $B \in \B_{S \times G}$
	$$
	\xi_A (B) = \xi (A \cap B).
	$$ 
	Equivalently, if $\xi := \sum_{\gamma_x \in \langle \xi \rangle} m_\xi(\gamma_x) \delta_{\gamma_x}$ then $\xi_A$ is given by the standard representation
	$$
	\xi_A = \sum_{\gamma_x \in \langle \xi \rangle \cap A} m_\xi(\gamma_x) \delta_{\gamma_x}.
	$$ To improve readability, in the following we will write $\xi_\Lambda$ instead of $\xi_{\Lambda \times G}$ for $\Lambda \in \B^0_S$.
\end{defi}

\begin{defi} The \textit{superposition} of two configurations $\sigma,\eta \in \mathcal{N}(S \times G)$ is defined as the particle configuration $\sigma \cdot \eta$ such that for every $B \in \B_{S \times G}$
	$$
	(\sigma \cdot \eta) (B) = \sigma(B) + \eta (B).
	$$
\end{defi}

In the particular cases in which $\sigma \in\mathcal{N}(\Lambda \times G)$ and $\eta\in \mathcal{N}(\Lambda^c \times G)$ for a certain $\Lambda \in \B^0_S$, the superposition $\sigma\cdot\eta$ can be thought of as a \textit{concatenation}. In these cases, 
the operations of restriction and superposition induce a natural identification between $\mathcal{N}(S \times G)$ and $\mathcal{N}(\Lambda \times G) \times \mathcal{N}(\Lambda^c \times G)$ for any given $\Lambda \in \B^0_S$. Indeed, the applications
$$
\begin{array}{rcl}
\mathcal{N}(S \times G)& \overset{r}{\longrightarrow}& \mathcal{N}(\Lambda \times G)\times \mathcal{N}(\Lambda^c \times G) \\  
\xi& \longmapsto &(\xi_{\Lambda},\xi_{\Lambda^c})
\end{array}
$$
and
$$
\begin{array}{rcl}	 
\mathcal{N}(\Lambda \times G)\times \mathcal{N}(\Lambda^c \times G) &\overset{s}{\longrightarrow}& \mathcal{N}(S \times G)\\
(\sigma,\eta) &\longmapsto& \sigma\cdot \eta
\end{array}
$$ 
are bijections and have each other as their respective inverse.

\subsubsection{Measurable structure}\label{sec:ms}

The space $\mathcal{N}(S \times G)$ is endowed with a measurable structure by considering the $\sigma$-algebra $\F$ generated by the counting events, i.e. 
\begin{equation}\label{salgebra}
\F \;=\; \sigma\left( \bigl\{ \xi \in \mathcal{N}(S \times G) : \xi(B) = k \bigr\} : k \in \N_0 \text{ and } B \in \B^0_{S\times G} \right).
\end{equation} 
Alternatively, if one considers for every $B \in \B_{S\times G}$ the counting variable 
$$
\begin{array}{rcl}
N_B : \mathcal{N}(S \times G) &\to& \N_0\\
N_B(\eta) &=& \eta(B),
\end{array}
$$
then 
$$
\F = \sigma\left( N_B : B \in \B^0_{S\times G}\right).
$$
More generally, for any $A \in \B_{S \times G}$ the $\sigma$-\textit{algebra $\F_A$ of events occurring in }$A$ is defined as the one generated by the counting events inside $A$, i.e.
\begin{eqnarray*}
	\F_A &:=& \sigma\left( \bigl\{ \xi \in \mathcal{N}(S \times G) : \xi(B) = k \bigr\} : k \in \N_0 \text{ and } B \in \B^0_A \right)\\
	&\;=& \sigma\left( N_B : B \in \B^0_A\right).
\end{eqnarray*}

The case $A=\Lambda \times G$ for $\Lambda \in \B^0_S$ is of particular relevance. First, we make the following important observation.

\begin{obs} The identification between $\mathcal{N}
	(S \times G)$ and $\mathcal{N}(\Lambda \times G) \times \mathcal{N}(\Lambda^c \times G)$ defined above is in fact a measurable isomorphism if the spaces are endowed with the $\sigma$-algebras $\F$ and $\F_{\Lambda \times G} \otimes \F_{\Lambda^c \times G}$, respectively.
\end{obs}

Further, we introduce the natural notions of local events and observables.

\begin{defi} $\,$
	\begin{enumerate}
		\item [i.] A function $f: \mathcal{N}(S\times G) \to \R$ is called a \textit{local observable} if there exists $\Lambda \in \B^0_S$ such that $f$ is $\F_{\Lambda \times G}$-measurable, i.e. if $f(\sigma)=f(\eta)$ whenever $\sigma_{\Lambda} = \eta_{\Lambda}$.  
		\item [ii.] An event $A \in \F$ is called \textit{local} if $\mathbbm{1}_A$ is a local observable, i.e. if $A\in\bigcup_{\Lambda \in \B^0_S}\F_{\Lambda \times G}$.
		\item [iii.] An event $A \in \F$ is called $\sigma$-\textit{local} if it is the countable union of local events.
	\end{enumerate}
\end{defi}

\subsubsection{Topological structure}

Physically, two configurations are close whenever inside some large compact set each configuration is a slight deformation of the other. This means that each particle inside this compact set of one configuration can be matched to a neighboring particle of the other and vice versa. The precise definition is as follows.

\begin{defi} \ 
	
	\begin{itemize}
		\item[i.]
		Given $\delta > 0$ and $\xi,\eta \in \mathcal{N}(S \times G)$ we say that $\xi$ is $\delta$-\textit{embedded} in $\eta$ if there exists an injective application $p:[\xi] \to [\eta]$ such that $d\left( \pi_{S\times G}(\gamma_x,i) , \pi_{S\times G}(p(\gamma_x,i))\right) < \delta$ for all $(\gamma_x,i) \in [\xi]$, with $\pi_{S\times G} : (S\times G) \times \N \to S \times G$ the projection onto $S \times G$. We denote it by $\xi \preceq_{\delta} \eta$.
		
		\item[ii.] Given a particle configuration $\xi \in \mathcal{N}(S \times G)$, a compact set $K \subseteq S \times G$ and $\delta > 0$, the $(K,\delta)$-neighborhood of $\xi$ is the set
		$$
		(\xi)_{K,\delta} = \bigl\{ \eta \in \mathcal{N}(S \times G) : \xi_K \preceq_\delta \eta \text{ and }\eta_K \preceq_\delta \xi \bigr\}.
		$$
	\end{itemize}
\end{defi}

The topology of the configuration space is the one defined by these neighborhoods.
\begin{defi} The \textit{vague topology} on $\mathcal{N}(S \times G)$ is the topology generated by the basis
	$$
	\mathfrak{B} = \bigl\{ (\xi)_{K,\delta} : \xi \in \mathcal{N}(S \times G) , K \subseteq S \times G\text{ compact and }\delta > 0\bigr\}.
	$$
\end{defi}

A number of observations are in order.

\begin{obs}\label{obs:r1}\ 
	
	\begin{itemize}
		\item[(a)] $\mathcal{N}(S \times G)$ admits a metric consistent with the vague topology, under which it is a separable metric space. It is also complete whenever $G$ is compact. 
		
		\item[(b)] The $\sigma$-algebra $\F$ defined in \eqref{salgebra} is actually the Borel $\sigma$-algebra corresponding to the vague topology on $\mathcal{N}(S \times G)$.
		
		\item[(c)]  The vague topology is usually defined as the one generated by the neighborhoods
		$$
		(\xi)_{f_1,\ldots, f_n,\delta} \;=\; \bigl\{ \eta \in \mathcal{N}(S \times G) : \left|\xi(f_i)
		-\eta(f_i)\right|<\delta\;,\; i=1,\ldots,n\bigr\}
		$$
		for $\delta>0$ and $f_1,\ldots, f_n$ continuous functions vanishing outside some compact set. This definition is equivalent to ours, but we will not use it in the sequel.
	\end{itemize}
\end{obs}

\subsubsection{Point processes on $S \times G$}

We call any random element of $\mathcal{N}(S \times G)$ a \textit{point process} on $S \times G$. Every point process $X$ on $S \times G$ is characterized by its distribution $P_X$, which is a probability measure on $\mathcal{N}(S \times G)$; the original measure space on which the process is defined plays no role. In the sequel, we will study convergence of point processes on $S \times G$ and, in general, of probability measures on $\mathcal{N}(S \times G)$. Besides the well-known notion of weak convergence, in our work we will also consider the notion of local convergence, which we define now.

\begin{defi}\label{localconvergence}\  
	
	\begin{enumerate} 
		\item [i.] A sequence $(\mu_n)_{n \in \N}$ of probability measures on $\mathcal{N}(S \times G)$ converges \textit{locally} to a probability measure $\mu$ on $\mathcal{N}(S \times G)$ if 
		$$
		\lim_{n \rightarrow +\infty}\mu_n(f) = \mu(f)
		$$ for every bounded local function $f:\mathcal{N}(S \times G) \to \R$. We denote this by $\mu \overset{loc}{\longrightarrow} \mu$. 
		
		\item [ii.] Likewise, a sequence $(X_n)_{n \in \N}$ of point processes on $S \times G$ converges locally to a point process $X$ on $S \times G$ if $P_{X_n} 
		\overset{loc}{\longrightarrow}P_X$. We denote this by $X_n \overset{loc}{\longrightarrow} X$. 
	\end{enumerate}
\end{defi}

In the general setting local observables need not be continuous, so that local and weak notions of convergence may not coincide in general. However, since uniformly continuous functions can be approximated arbitrarily well in the supremum norm by local observables, we have that local convergence is always stronger than weak convergence. Both notions coincide in fact whenever $S$ and $G$ are both countable and discrete. 

\subsection{Gas models}

Our models consist of three ingredients: (1) a configuration space, (2) an underlying ``free'' measure on the configuration space, and (3) a notion of interaction encoded in Hamiltonians. The first ingredient was the object of the preceding subsection; here we discuss the remaining two and give the general definition of gas model.

\subsubsection{Poisson processes on $S \times G$ and the free measure}

In all gas models the underlying free measure will be a Poisson distribution on $\mathcal{N}(S \times G)$ for some appropriate intensity. We give the definition of Poisson distribution below.

\begin{defi} Let $\nu$ be a measure on $(S \times G,\B_{S \times G})$ of $S$-locally finite allocation [i.e. such that $\nu(\Lambda \times G) < +\infty$ for every $\Lambda \in \mathcal{B}^0_S$]. 
	\begin{enumerate}
		\item [i.] The \textit{Poisson distribution with intensity} $\nu$ is the unique measure $\pi^\nu$ on $\mathcal{N}(S \times G)$ which satisfies
		$$
		\pi^\nu ( \{ \xi \in \mathcal{N}(S \times G) : \xi(B_i) = k_i \text{ for all }i=1,\dots,n \} ) = \prod_{i=1}^n \frac{e^{-\nu(B_i)} \left(\nu(B_i)\right)^{k_i}}{k_i!}
		$$ for all $k_1,\dots,k_n \in \N_0$, disjoint $B_1,\dots,B_n \in \B^0_{S \times G}$ and $n \in \N$.
		\item [ii.] A point process $X$ is called a \textit{Poisson process} with intensity $\nu$ if it is distributed according to $\pi^\nu$, i.e. for every finite collection of disjoint sets $B_1,\dots,B_n \in \B^0_{S \times G}$ the random variables $X(B_1),\dots,X(B_n)$ are independent and Poisson-distributed with respective means $\nu(B_1),\dots,\nu(B_n)$.
	\end{enumerate}
\end{defi}

\subsubsection{Hamiltonians}

In this article we aim at treating general gas models, including also contour ensembles. The latter are not amenable to a standard Gibbsian description. Rather, the energy cost of particle configurations in these models is described in terms of a Hamiltonian prescription, according to the following formal definition.

\begin{defi} \ 
	
	\begin{itemize}
		\item[i.] A \textit{Hamiltonian prescription} on $\mathcal{N}(S\times G)$ is a family of measurable functions 
		$$
		H = \bigl\{H_{\Lambda|\eta} : \mathcal{N}(\Lambda \times G)  \to (-\infty,+\infty]\;: \Lambda\in 
		\B^0_S\,, \eta\in \mathcal{N}(S \times G)\bigr\}.
		$$ 
		The function $H_{\Lambda|\eta}$ is called the \textit{local Hamiltonian on $\Lambda$ with boundary condition $\eta$}. 
		\item[ii.] Given a Hamiltonian prescription $H$, $\Lambda\in \B^0_S$ and a configuration $\eta\in \mathcal{N}(\Lambda^c \times G)$, the \textit{energy leap in $\Lambda$ relative to $\eta$} is the function $\Delta E_{\Lambda|\eta} : \Lambda \times G \to (-\infty,+\infty]$ defined by 
		$$
		\Delta E_{\Lambda|\eta} (\gamma_x) \;=\; \left\{\begin{array}{ll}H_{\Lambda|\eta}( \eta_{\Lambda \times G} + \delta_{\gamma_x} ) - H_{\Lambda|\eta}(\eta_\Lambda) & \text{ if $H_{\Lambda|\eta}(\eta_\Lambda) < +\infty$}\\ \\ +\infty & \text{ otherwise\;.}\end{array}\right.
		$$
	\end{itemize}
\end{defi}

The local Hamiltonian $H_{\Lambda|\eta}$ measures the energy cost 
of inserting configurations in $\Lambda$ if surrounded by the configuration $\eta$. The energy leap $\Delta E_{\Lambda|\eta}$ represents the energy cost of placing an additional particle $\gamma_x$ inside $\Lambda$ when in the presence of the configuration $\eta$. In particular, we may define the \textit{impact relation} $\rightharpoonup$ on $S \times G$ by the rule 
\begin{equation}\label{eq:rob1}
	\tilde{\gamma}_y \rightharpoonup \gamma_x \Longleftrightarrow \exists\,\, \Lambda \in \B^0_S \text{ and }\eta \in \mathcal{N}(S\times G) \text{ with }\Delta E_{\Lambda|\eta}(\gamma_x) \neq \Delta E_{\Lambda|\eta + \delta_{\tilde{\gamma}_y}} (\gamma_x).
	\end{equation}
If $\tilde{\gamma}_y \rightharpoonup \gamma_x$ we say that $\tilde{\gamma}_y$ has an impact on $\gamma_x$. This relation is not necessarily symmetric. 

\subsubsection{Definition of gas model}

Every gas model is defined by a pair $(\nu,H)$, where $\nu$ is an intensity measure and $H$ a Hamiltonian prescription. The former describes how particles would be
distributed if there would be no interaction among them, and the latter specifies this interaction. Our aim of considering general point processes 
forces us to list a relatively long list of assumptions. As illustrated below, these assumptions are naturally satisfied by usual examples.

\begin{defi} \label{assump} A \textit{gas model} on $S \times G$ is a pair $(\nu,H)$ verifying the following conditions:
	\begin{enumerate}
		\item \textit{$S$-locally finite allocation}. For all $\Lambda \in \mathcal{B}^0_S$ the measure $\nu$ satisfies $\nu(\Lambda \times G) < +\infty$.
		\item \textit{Diluteness condition:} $H_{\Lambda|\eta}( \emptyset ) < +\infty$ for every $\Lambda \in \B^0_S$ and $\eta \in \mathcal{N}(S \times G)$.
		\item \textit{Existence of infinite-volume energy leap function:} The limit 
		$$
		\Delta E_\eta (\gamma_x) := \lim_{\Lambda \nearrow S} \Delta E_{\Lambda|\eta}(\gamma_x)
		$$ exists for all $\gamma_x \in S\times G$ and $\eta \in \mathcal{N}(S \times G)$. 
		\item \textit{Bounded energy loss and allowance of particles:}
		$$
		-\infty < \Delta E := \inf_{\Lambda \in \B^0_S} \left[ \inf_{\substack{ \gamma_x \in \Lambda \times G \\ \eta \in \mathcal{N}(S \times G)}} \Delta E_{\Lambda|\eta} (\gamma_x) \right] < +\infty.
		$$
		
		\item \textit{Integrable interaction range:} Let the \textit{interaction range} of $B \in \B_{S \times G}$ be the set
		$$
		I(B)= \{ \tilde{\gamma}_y \in S \times G : \exists\,\, \gamma_x \in B \text{ such that } \tilde{\gamma}_y \rightharpoonup \gamma_x \}\;.
		$$ Then $I(B)$ is measurable and $\nu\bigl(I(\Lambda \times G)\bigr) < +\infty$ for each $\Lambda \in \mathcal{B}^0_S$.
		
		\item \textit{Measurability of local Hamiltonians:} Given $\Lambda \in \B^0_S$ and $\gamma_x \in S \times G$,
		\begin{enumerate}
			\item [i.] The application $(\xi,\eta)\to H_{\Lambda|\eta}(\xi)$ is $(\F_{\Lambda \times G} \otimes \F_{(\Lambda^c \times G) \cap I(\Lambda \times G)})$-measurable.
			\item [ii.] The application $\eta \mapsto \Delta E_{\eta} (\gamma_x)$ is $\F_{I(\gamma_x)}$-measurable.
		\end{enumerate}
		
	\end{enumerate}
\end{defi}

These conditions are satisfied by all physical systems of interest we know of (although for some discrete systems, like the Ising model for example, one may need to consider an alternative lattice gas representation for these to hold), with the exception of condition (4). This condition is violated, for example, by interactions of Lennard-Jones type. Indeed, if one considers a ring of particles of a radius for which the L-J potential is negative, then the addition of a particle at the center of the ring would lead to an energy leap that becomes arbitrarily low with the (potentially unbounded) number of particles in the ring, thus yielding $\Delta E = -\infty$. This suggests that, except for systems with purely nonnegative interactions, the validity of the leftmost inequality in (4) is tantamount to the existence of some sort of hard-core requirement preventing arbitrarily large amounts of particles inside bounded regions.

\subsubsection{Gas kernels}

Every gas model defines a family of probability measures on $\mathcal{N}(S \times G)$, called \textit{gas kernels}, which describe the local behavior of the system in bounded volumes. We introduce this family of gas kernels below.

\begin{defi}\label{defibgd} 
	The \textit{gas kernel of }$(\nu,H)$ \textit{on the volume} $\Lambda \in \mathcal{B}^0_S$ \textit{with boundary condition} $\eta \in \mathcal{N}( S \times G)$ is the probability measure $\mu_{\Lambda|\eta}$ on $\mathcal{N}(S \times G)$ given by 
	\begin{equation}\label{Gibbs1}
	\mu_{\Lambda|\eta} = \omega_{\Lambda|\eta} \times \delta_{\eta_{\Lambda^c}}
	\end{equation} where we make the identification $\mathcal{N}(S \times G)=\mathcal{N}( \Lambda \times G) \times \mathcal{N}(\Lambda^c \times G)$ and $\omega^\eta_\Lambda$ denotes the probability measure on $\mathcal{N}(\Lambda \times G)$ defined through the relation
	$$
	d\omega_{\Lambda|\eta} = \frac{e^{-H_{\Lambda|\eta}}}{Z_{\Lambda|\eta}} d\pi^\nu_\Lambda
	$$ with $\pi^\nu_\Lambda$ denoting the Poisson distribution on $\mathcal{N}(\Lambda \times G)$ with intensity $\nu_{\Lambda \times G}$ and
	$$
	Z_{\Lambda|\eta}=\displaystyle{\int_{\mathcal{N}(\Lambda \times G)} e^{-H_{\Lambda|\eta}(\sigma)} d\pi^\nu_\Lambda(\sigma)}
	$$ being the normalization constant. Notice that by assumptions (1)-(2) in Definition \ref{assump} we have
	$$
	Z_{\Lambda|\eta} \geq e^{-H_{\Lambda|\eta}(\emptyset)}\pi^\nu ( N_{\Lambda \times G} = 0 ) = e^{- (\nu(\Lambda \times G) + H_{\Lambda|\eta}(\emptyset)) } > 0
	$$ so that $\omega_{\Lambda|\eta}$ is well defined.
\end{defi}



The measures $\mu_{\Lambda|\eta}$ describe the local behavior of the system inside the volume $\Lambda$ when the configuration outside $\Lambda$ is fixed as $\eta$. The true objects of interest for us are, however, the possible local limits of these along suitable boundary conditions. 

\begin{defi}\label{gas1} Let $(\nu,H)$ be a gas model.
		\begin{enumerate}
			\item [i.] We say that a configuration $\eta \in \mathcal{N}(S \times G)$ has \text{finite $H$-interaction range} whenever $\eta(I(\Lambda \times G)) < +\infty$ for every $\Lambda \in \B^0_S$, i.e. if $\eta$ has only finitely many particles interacting with those in any given bounded volume. 
			\item [ii.] A probability measure $\mu$ on $\mathcal{N}(S \times G)$ is called a $(\nu,H)$-\textit{gas measure} if there exist $\eta \in \mathcal{N}(S \times G)$ with finite $H$-interaction range and $(\Lambda_n)_{n \in \N} \subseteq \B^0_S$ with $\Lambda_n \nearrow S$ such that  
			\begin{equation}\label{eq:r10}
			\mu_{\Lambda_n|\eta} \overset{loc}{\longrightarrow} \mu.
			\end{equation}
		\end{enumerate}
	\end{defi}
	
	By condition (5) in Definition \ref{assump}, configurations with  an infinite $H$-interaction range are not physically admissible for the system. Thus, allowing these as boundary conditions in \eqref{eq:r10} may lead to pathological limits of no physical meaning. Indeed, it is not hard to device examples in which $\delta_\emptyset$, the $\delta$-measure on the empty configuration, can be obtained as the local limit in \eqref{eq:r10} along a boundary condition with an infinite $H$-interaction range. This is why in the definition of gas measure we disregard this type of boundary conditions.
	
	\subsubsection{Gibbsian gas models}
	
	Gas measures possess a clear physical interpretation for the particular case of Gibssian models. These models are defined by Hamiltonian prescriptions in which local interactions do not depend on the particular volume under consideration. 
	
	\begin{defi} We say that a model $(\nu,H)$ is \textit{Gibbsian} if the Hamiltonian prescription $H$ satisfies the consistency property 
		\begin{equation}\label{eq:r20}
		H_{\Lambda|\eta}(\sigma) = H_{\Delta|\sigma_{\Lambda} \cdot \eta_{\Lambda^c }}(\sigma_{\Delta}) + H_{(\Lambda \setminus \Delta)| \emptyset_{\Lambda}\cdot \eta_{\Lambda^c }}(\sigma_{\Lambda \setminus \Delta})
		\end{equation} 
		for every $\Delta \subseteq \Lambda \in \B^0_S$, $\sigma \in \mathcal{N}(\Lambda \times G)$ and $\eta \in \mathcal{N}(S\times G)$.
	\end{defi}
	
	Notice that condition \eqref{eq:r20} implies that, as we anticipated, the energy leap $\Delta E_{\Lambda|\eta} (\gamma_x)$ does not depend on $\Lambda$. In particular, assumption (3) in Definition \ref{assump} immediately holds for Hamiltonians verifying \eqref{eq:r20}. Furthermore, condition \eqref{eq:r20} yields gas kernels satisfying a consistency relation of the form
		\begin{equation}\label{eq:r21}
		\mu_{\Lambda|\eta} \;=\; \int \mu_{\Delta|\xi} \, d\mu_{\Lambda|\eta}(\xi)
		\end{equation}
		for every pair of volumes $\Delta \subseteq \Lambda \in \B^0_S$ and $\eta \in \mathcal{N}(S\times G)$.
	A family of gas kernels satisfying \aref{eq:r21} is called a \textit{specification}.
	The infinite-volume measures relevant to Gibbsian models are usually introduced through the $\Lambda\to S$ version of \aref{eq:r21}.
	
	\begin{defi}\label{Gibbs2} 
		Let $(\nu,H)$ be a Gibbsian model. A probability measure $\mu$ on $\mathcal{N}( S \times G)$ is called a \textit{Gibbs measure} for $(\nu,H)$ if for every $\Lambda \in \B^0_S$
		\begin{equation}\label{eq:r22}
		\mu = \int \mu_{\Lambda|\eta} \,d\mu(\eta).
		\end{equation}	
	\end{defi}
	
	Thus, Gibbs measures are precisely those infinite-volume measures which are consistent (in the sense of \eqref{eq:r21}) with the local description of the model given by its specification. Hence, we may think of them as the measures describing the global states of our system. 
	The relation between Definitions \ref{Gibbs2} and \ref{gas1} follows from the fact that, for a large class of Gibssian models (see \cite{SJS} for details), any local limit of gas kernels as put in \eqref{eq:r10} is in fact a Gibbs measure for $(\nu,H)$ and, furthermore, all extremal Gibbs measures can be obtained in this way.
	
	In practice, almost every model of physical interest is Gibbsian. There is, however, one important exception: contour ensembles. These are of particular relevance since they constitute one of the main tools for studying Gibbsian systems in the low-temperature (or condensation) regime. We give further discussion on Gibbsian and non-Gibbsian models in the following section.   
	
\subsection{Examples of gas models}\label{secexamples}
	
	We now present some examples of models to illustrate the definitions of the previous section. Later in Section \ref{secapps} we will also use these as ground for applications of our results. 
	
	\subsubsection{Models given by an interaction potential}
	
	The typical way in which Hamiltonians satisfying \eqref{eq:r20} are specified is via an \textit{interaction potential}, i.e. a family $\Phi = (\Phi^{(n)})_{n \in \N}$ of symmetric functions $\Phi^{(n)} : (S \times G)^n \rightarrow (-\infty,+\infty]$ subject to appropriate measurability and summability requirements so that the local Hamiltonians
	\begin{equation}\label{eq:hampot}
	H_{\Lambda|\eta}(\sigma):= \sum_{n\ge 1\atop m\ge 0} \frac{1}{n!m!}
	\sum_{(\gamma_x^{(1)},i_1),\dots,(\gamma_x^{(n)},i_n) \in [\sigma]
		\atop (\tilde{\gamma}_y^{(1)},j_1),\dots, (\tilde{\gamma}_y^{(m)},j_m) \in [\eta_{\Lambda^c}]}
	\Phi^{(n+m)}\bigl(\gamma_x^{(1)},\dots,\gamma_x^{(n)},\tilde{\gamma}_y^{(1)},\dots,\tilde{\gamma}_y^{(m)}\bigr)
	\end{equation}
	are well-defined and satisfy all the pertinent conditions in Definition \ref{assump}. The resulting Hamiltonian prescription is said to be \textit{specified} by $\Phi$ and will be often denoted by $H^\Phi$. The function $\Phi^{(n)}$ is called the $n$-\textit{body interaction} of the potential $\Phi$. 
	
	We say that a model $(\nu,H)$ is given by an interaction potential if $H=H^\Phi$ for some $\Phi$. It follows from \eqref{eq:hampot} that any such model is Gibbsian. Furthermore, we have that:
	
	\begin{enumerate}
		\item [$\bullet$] $H^\Phi_{\Lambda|\eta}(\emptyset) = 0$ for any $\Lambda \in \B^0_S$ and $\eta \in \mathcal{N}(S \times G)$.
		\item [$\bullet$] For any $\gamma_x \in S \times G$ and $\eta \in \mathcal{N}(S \times G)$ the energy leap $\Delta E^\Phi_{\eta} (\gamma_x)$ takes the form
		\begin{equation}
		\label{eq:phileap}
		\Delta E^\Phi_{\eta} (\gamma_x) = \sum_{m \geq 0} \frac{1}{m!} \sum_{(\tilde{\gamma}_y^{(1)},j_1),\dots, (\tilde{\gamma}_y^{(m)},j_m) \in [\eta]}
		\Phi^{(m+1)}\bigl(\gamma_x,\tilde{\gamma}_y^{(1)},\dots,\tilde{\gamma}_y^{(m)}\bigr).
		\end{equation} 
		\item [$\bullet$] The bounded energy loss condition in (4) of Definition \ref{assump} in this case reduces to the existence of a constant $C > 0$ such that 
		$$
		\sum_{m \geq 0} \sum_{(\tilde{\gamma}_y^{(1)},j_1),\dots, (\tilde{\gamma}_y^{(m)},j_m) \in [\eta]}
		\Phi^{(m+1)}\bigl(\gamma_x,\tilde{\gamma}_y^{(1)},\dots,\tilde{\gamma}_y^{(m)}\bigr) \geq -C
		$$ for every $\gamma_x \in S \times G$ and $\eta \in \mathcal{N}(S \times G)$. This condition is well-known and standard in the study of gas systems (see {\cite[Section~1.2]{DR}}).
	\end{enumerate} 
	
%
Models specified by an interaction potential are the most common among gas models. Below we give some examples.
	
	\subsubsection{The discrete Widom-Rowlinson model}\label{sec:dwr}
	
	It is a classical hardcore interaction model, first introduced by Lebowitz and Gallavotti in \cite{L}. It involves particles of two types, say $(+)$-particles and $(-)$-particles, located at the sites of the discrete lattice $\Z^d$ for $d \geq 1$. The interaction between particles allows at most one particle site and forbids any two particles of different type from being within a certain fixed distance $k \in \N$ of each other. The corresponding gas model is defined by the following ingredients: 
	\begin{enumerate}
		\item [$\bullet$] Location space $S=\Z^d$ and spin set $G= \{+,-\}$.
		\item [$\bullet$] Intensity measure given by 
		\begin{equation}
		\label{eq:dwrint}
		\nu = \lambda_+ \cdot c_{\Z^d} \times \delta_+ + \lambda_- \cdot c_{\Z^d} \times \delta_-,
		\end{equation} where $\lambda_+,\lambda_- > 0$ are two fixed parameters known as the \textit{fugacities} of $(\pm)$-particles respectively and $c_{\Z^d}$ denotes the counting measure on $\Z^d$. 
		\item [$\bullet$] Hamiltonian prescription $H$ specified by the potential $\Phi = \Phi^{(2)}$ given by
		\begin{equation}\label{eq:dwr2}
		\Phi^{(2)}(\gamma_x,\tilde{\gamma}_y):= \left\{ \begin{array}{ll} +\infty &\text{if }x=y \\ \\ +\infty & \text{if }0< \|x-y\|_\infty \leq k \text{ and }\gamma \neq \tilde{\gamma}\\ \\ 0 &\text{otherwise.}\end{array}\right.
		\end{equation}
	\end{enumerate} Alternatively, one could define the model by considering instead the pair $(c_{\Z^d \times \{+,-\}},\tilde{H})$, where $c_{\Z^d \times \{+,-\}}$ is the counting measure on $\Z^d \times \{+,-\}$ and $\tilde{H}$ is specified by the potential $\tilde{\Phi} = (\tilde{\Phi}^{(1)},\tilde{\Phi}^{(2)})$ with $\tilde{\Phi}^{(2)}$ as in \eqref{eq:dwr2} and 
	$$
	\tilde{\Phi}^{(1)}(\gamma_x) = \left\{\begin{array}{ll}-\log \lambda_+ & \text{if $\gamma=+$}\\ \\ -\log \lambda_- &\text{if $\gamma=-$.}\end{array}\right.
	$$ Both representations are \textit{equivalent} in the sense that they produce the same gas kernels. Nevertheless, for our analysis it will be more convenient to adopt the first representation. The reason for this choice will be explained later in Section \ref{secapps}. We adopt this representation also in the remaining examples.
	 
	
	\subsubsection{The continuum Widom-Rowlinson model} 
	In the continuum version of the model, particles are now located throughout the entire Euclidean space $\R^d$ and the interaction forbids particles of different type from being within a certain distance $r > 0$ of each other. It was originally introduced by Widom and Rowlinson in \cite{WR} and later studied in \cite{CCK,RU}.
	Its formal ingredients are:
	\begin{enumerate}
		\item [$\bullet$] Location space $S=\R^d$ and spin set $G= \{+,-\}$.
		\item [$\bullet$] Intensity measure 
		$$
		\nu:= \lambda_+ \cdot \mathcal{L}^d \times \delta_+ + \lambda_- \cdot \mathcal{L}^d \times\delta_-,
		$$ where $\mathcal{L}^d$ is the Lebesgue measure on $\R^d$. 
		\item [$\bullet$] Hamiltonian prescription $H$ specified by the potential $\Phi=\Phi^{(2)})$ given by  
		\begin{equation}\label{eq:cwr2}
		\Phi^{(2)}(\gamma_x,\tilde{\gamma}_y):= \left\{ \begin{array}{ll} +\infty & \text{if } \|x-y\|_\infty \leq r \text{ and }\gamma \neq \tilde{\gamma}\\ \\ 0 &\text{otherwise.}\end{array}\right.
		\end{equation} 
	\end{enumerate} Notice that the first term in \eqref{eq:dwr2} excluding multiple particles in one site is now missing from \eqref{eq:cwr2}. This is because the Poisson distribution $\pi^{\nu}$ already assigns zero probability to configurations with more than one particle per site and so this term becomes unnecessary. 
	
	\subsubsection{The Widom-Rowlinson model with generalized interactions}
	
	Several generalizations of the Widom-Rowlinson model are worth looking into. One interesting possibility is to consider a model in which nearby pairs of particles of opposite type are not necessarily forbidden, but merely discouraged, and also  intra-species repulsion terms are included. Such generalization is defined through  decreasing functions $h,j_\pm: \R^+ \rightarrow [0,+\infty]$ with bounded support, by replacing \eqref{eq:dwr2}-\eqref{eq:cwr2} with the $2$-body interaction  
	$$
	\Phi^{(2)} (\gamma_x,\tilde{\gamma}_y)= \left\{\begin{array}{ll}h(\|x-y\|_\infty) & \text{ if $\gamma \neq \tilde{\gamma}$} \\ \\ j_-(\|x-y\|_\infty) & \text{ if $\gamma=\tilde{\gamma}=-$}\\ \\j_+(\|x-y\|_\infty) & \text{ if $\gamma=\tilde{\gamma}=+$}.
	\end{array}\right.
	$$
	We call $(h,j_-,j_+)$ the \textit{repulsion vector}. 
	The original continuum Widom-Rowlinson model is obtained by setting $h:= (+\infty) \mathbbm{1}_{[0,r]}$ and $j_\pm \equiv 0$. The discrete version corresponds to the same choice of $h$ and $j_\pm := (+\infty) \mathbbm{1}_{\{0\}}$. We refer to \cite{GH} where these type of generalizations were investigated.
	
	\subsubsection{The thin rods model in $\Z^2$}\label{sec:dtrm}
	
	Given $k \in \N$ we consider a system of hard rods in $\R^2$ of zero width and length $2k$ whose centers are located at the sites of $\Z^2$. Each rod has an orientation specified by an angle $\gamma \in [0,\pi)$ with respect to the $x$-axis, and the interaction forbids any two rods to intersect. More precisely, if for $r > 0$ we set
	$$
	L_\gamma^r \;:=\; \{ t \cdot (\cos\gamma, \sin \gamma) : t \in [-r,r] \}
	$$ 
	then the thin rods model in $\Z^2$ is defined by:
	\begin{enumerate}
		\item [$\bullet$] Location space $S=\Z^2$ and spin set $G= [0,\pi)$.
		\item [$\bullet$] Intensity measure $\nu := \lambda \cdot c_{\Z^2} \times \rho$, where $\lambda > 0$ is called the \textit{fugacity} of rods and $\rho$ is a probability measure on $G$ called the \textit{orientation measure}. 
		\item [$\bullet$] Hamiltonian prescription $H$ specified by the potential $\Phi=\Phi^{(2)}$ given by   
		\begin{equation}\label{tru}
		\Phi^{(2)}(\gamma_x,\tilde{\gamma}_y) := \left\{ \begin{array}{ll} +\infty &\text{if }(L^k_\gamma + x) \cap (L^k_{\tilde{\gamma}} + y) \neq \emptyset\\ 0 &\text{otherwise.}\end{array}\right.
		\end{equation}
	\end{enumerate}
	Of particular interest to us is the case when the orientation measure is given by
	\begin{equation} \label{eq:om}
	\rho = p \delta_{0} + (1-p)\delta_{\frac{\pi}{2}}
	\end{equation} for some $p \in (0,1)$. This model is identical to the discrete Widom-Rowlinson model, with the exception of an additional repulsion term between particles of the same type. Indeed, \mbox{by identifying} $G$ with $\{+.-\}$ we have that $\nu$ equals \eqref{eq:dwrint} for $\lambda_+:=p\lambda$ and $\lambda_-:= (1-p)\lambda$, while $\Phi^{(2)}$ in \eqref{tru} can be rewritten as $\Phi^{(2)}=\Phi^{(2)}_{WR} + \Phi_*^{(2)}$, where $\Phi^{(2)}_{WR}$ is as in \eqref{eq:dwr2} and 
		\begin{equation} \label{eq:w}
		\Phi_*^{(2)}(\gamma_x,\tilde{\gamma}_y) = \left\{\begin{array}{ll} +\infty & \text{ if $\gamma=\tilde{\gamma}=+$ and $|x_1-y_1|\leq k$}\\ \\ +\infty & \text{ if $\gamma=\tilde{\gamma}=-$ and $|x_2-y_2|\leq k$}\\ \\ 0 & \text{otherwise}.
		\end{array}\right.
		\end{equation} 
	We call this particular system the \textit{nematic thin rods model}. More details about this model can be found in \cite{DG,GD} and references therein.
	
	\subsubsection{The thin rods model in $\R^2$} Similar to the previous model, the rod centers are now located at arbitrary points of $\R^2$ and the rod lengths are $2r$ for some fixed $r > 0$ which is not necessarily an integer. The model is formally defined by:
	\begin{enumerate}
		\item [$\bullet$] Location space $S=\R^2$ and spin set $G= [0,\pi)$.
		\item [$\bullet$] Intensity measure $\nu = \lambda\cdot \mathcal{L}^2 \times \rho$, where $\lambda > 0$ and $\rho$ is the orientation measure.
		\item [$\bullet$] Hamiltonian prescription $H$ specified by the potential $\Phi=\Phi^{(2)}$ where 
		$$
		\Phi^{(2)}(\gamma_x,\tilde{\gamma}_y) := \left\{ \begin{array}{ll} +\infty &\text{if }(L^r_\gamma + x) \cap (L^r_{\tilde{\gamma}} + y) \neq \emptyset\\ 0 &\text{otherwise.}\end{array}\right.
		$$ 
	\end{enumerate}
	For $\rho$ as in \eqref{eq:om}, the model is equivalent to the continuum Widom-Rowlinson model since $\pi^{\nu}$ assigns zero probability to configurations for which $\Phi_*^{(2)}$ in \eqref{eq:w} would be nonzero. Thus, both the Widom-Rowlinson model and the nematic thin rods model have the same continuum version. We refer to \cite{BKL} where the model with a finite number of orientations was studied.
	
	\subsubsection{The Peierls contours model} To conclude, we show an example of a model which is not Gibbsian, but still fits into our framework: the Peierls contours model. It was first presented by Peierls in \cite{Peierls36} to study the Ising model at low temperature (see also \cite{Griffiths64}). For simplicity, we focus only in the $2$-dimensional case. 

The Ising model is the lattice system on the configuration space $\{-1,+1\}^{\Z^2}$ defined by the set of finite-volume specifications $\{ \mu^I_{\Lambda|\eta} : \Lambda \in \B^0_{\Z^2} , \eta \in \{-1,+1\}^{\Z^2}\}$ given by
$$
\mu^{I}_{\Lambda|\eta}(\sigma) = \frac{\mathbbm{1}_{\{\sigma_{\Lambda^c} \equiv \eta_{\Lambda^c}\}}}{Z^I_{\Lambda|\eta}} e^{-H^I_{\Lambda|\eta}(\sigma_\Lambda)},
$$ where 
\begin{equation}
\label{eq:ising}
H^I_{\Lambda|\eta}(\sigma_\Lambda):= -\frac{\beta}{2}\sum_{\substack{x,y \in \Lambda \\
		\|x-y\|_2=1}} \sigma(x)\sigma(y) - \beta \sum_{ \substack{x \in \Lambda, y \notin \Lambda \\ \|x-y\|_2=1}} \sigma(x)\eta(y)
\end{equation} for a fixed parameter $\beta > 0$ known as the \textit{inverse temperature}. Let us observe that, since the interaction has range one, the boundary condition $\eta$ is involved in \eqref{eq:ising} only through its values on the external boundary $\partial \Lambda$ of $\Lambda$ defined as
$$
\partial \Lambda := \{y \notin \Lambda^c : d_2(y,\Lambda)=1\}.
$$ Of particular interest are the boundary conditions $+$ and $-$, corresponding to $\eta(x)=+1$ and $\eta(x)=-1$ for all $x \in \Z^2$, respectively. It can be seen that the local limits
$$
\mu^+ := \lim_{\Lambda \nearrow \Z^2} \mu^I_{\Lambda|+} \hspace{2cm}\text{ and }\hspace{2cm}\mu^-:= \lim_{\Lambda \nearrow \Z^2} \mu^I_{\Lambda|-}
$$ both exist and constitute the unique extremal Gibbs measures of the model for a fixed $\beta$, in the sense of Definition \ref{Gibbs2} (see \cite{GHM} and references therein for details). If $\beta$ is such that $\mu^+$ and $\mu^-$ do not coincide, we say that a \textit{phase transition} occurs at \mbox{inverse temperature $\beta$.} Peierls showed the existence of a phase transition for all sufficiently large values of $\beta$ by considering the following geometric description of configurations in terms of contours.

We begin by fixing $+$ as the boundary condition and letting $\Delta \in \B^0_{\Z^2}$ be a square. Now, consider $\Z^2_*:= \Z^2 + (\frac{1}{2},\frac{1}{2})$, the dual lattice of $\Z_2$. Given an edge $e$ joining two neighboring sites in $\Z^2$, let $e_*$ denote the unique edge joining neighboring sites in $\Z^2_*$ which is orthogonal to $e$. We call $e_*$ the \textit{dual edge} of $e$. Furthermore, consider:
\begin{enumerate}
	\item [$\bullet$] $e(\overline{\Delta})$, the set of edges in $\Z^2$ with at least one endpoint in $\Delta$.
	\item [$\bullet$] $e_*(\overline{\Delta}):=\{e_* : e \in e(\overline{\Delta})\}$, the set of dual edges of $e(\overline{\Delta})$.
	\item [$\bullet$] $\Delta_*$, the set of sites in $\Z^2_*$ which are endpoints of edges in $e_*(\overline{\Delta})$. 
\end{enumerate} Given a configuration $\sigma$ satisfying the boundary condition $+$ outside $\Delta$, let $D_\sigma$ denote the set of dual edges $e_* \in e_*(\overline{\Delta})$ such that $e$ joins two sites $x,y$ with different spin, i.e. $\sigma(x)\sigma(y)=-1$. With a little work it is possible to show that the edges in $D_\sigma$ \mbox{join up to} form closed curves (which may contain loops). This set of curves can be decomposed into connected components $\gamma_1,\dots,\gamma_n$. We call any of these components $\gamma_i$ a \textit{contour}, and write $\Gamma_\sigma:=\{\gamma_1,\dots,\gamma_n\}$ for the set of contours of $\sigma$. It can be seen that the assignation $\sigma \mapsto \Gamma_\sigma$ is in fact a bijection: given a finite family $\Gamma$ of mutually disjoint contours contained in $e_*(\overline{\Delta})$, there exists a unique configuration $\sigma_\Gamma$ satisfying the boundary condition $+$ in $\Delta^c$ which has $\Gamma$ as its set of contours. Furthermore, if $|\gamma|$ denotes the number of edges in $\gamma$, then for any such $\sigma$ we have
\begin{equation}
\label{eq:iscr}
\mu^I_{\Delta|+}(\sigma) = \frac{1}{W_\Delta} e^{-2\beta \sum_{\gamma \in \Gamma_\sigma}|\gamma|},
\end{equation} where $W_\Delta$ is a normalizing constant depending solely on $\Delta$. Thus, \mbox{whenever $\Delta$ is a square,} with \eqref{eq:iscr} we obtain an alternative representation of $\mu^I_{\Delta|+}$ in terms of a \mbox{system of contours} interacting by exclusion. In the current framework of gas models, this system is defined by setting:

\begin{enumerate}
	\item [$\bullet$] The dual lattice $\Z^2_*$ as the location space $S$. 
	\item [$\bullet$] The set of contours rooted at the origin $0_*:=(\frac{1}{2},\frac{1}{2}) \in \Z^2_*$ as the spin set $G$. Here, we say that a contour $\gamma$ is \textit{rooted} at $x \in \Z^2_*$ if $x$ is the smallest site belonging to $\gamma$ (with respect to the lexicographical order). Thus, we interpret any $\gamma_x \in \Z^2_* \times G$ as the contour shape $\gamma$ rooted at $x$.
	\item [$\bullet$] The intensity measure $\nu$ given for each $\gamma_x \in S \times G$ by $\nu(\gamma_x) := e^{-\rojo{2}\beta|\gamma_x|}$.
	\item [$\bullet$] For each $\Lambda \in \B^0_{\Z^2_*}$ and $\eta \in \mathcal{N}(\Z^2_* \times G)$, the Hamiltonian $H_{\Lambda|\eta}$ specified as in \eqref{eq:hampot} but for the \textit{local} potential $\Phi_{\Lambda}=(\Phi^{(1)}_\Lambda,\Phi^{(2)})$ given by
	$$
	\Phi^{(2)}(\gamma_x,\tilde{\gamma}_x) = \left\{\begin{array}{ll}+\infty & \text{ if $\gamma_x \cap \tilde{\gamma}_y \neq \emptyset$} \\ \\ 0 & \text{ otherwise}\end{array}\right.
	$$ and 
	$$
	\Phi^{(1)}_{\Lambda}(\gamma_x) = \left\{\begin{array}{ll} +\infty & \text{ if $\gamma_x \cap (\Z^2_* - \Lambda) \neq \emptyset$} \\ \\ 0 & \text{ otherwise.}\end{array}\right.
	$$ The interaction term $\Phi^{(2)}$ is responsible for the exclusion among different contours, while the term $\Phi^{(1)}_\Lambda$ bans those contours which are not contained in $\Lambda$.
	\end{enumerate}
The resulting pair $(\nu,H)$ is called the \textit{Peierls contours model}. 
The main physical interest of this model lies in the fact that phase transitions in the Ising model at low temperatures can be understood in terms of the diluteness properties of gas measures for $(\nu,H)$. Indeed, if $\{ \mu^P_{\Lambda|\eta} : \Lambda \in \B^0_{\Z^2_*} , \eta \in \mathcal{N}(\Z^2_* \times G)\}$ is the family of kernels induced by $(\nu,H)$, then \eqref{eq:iscr} can be rewritten as 
	\begin{equation}
	\label{eq:iscr2}
	\mu^I_{\Delta|+}(\sigma)=\mu^P_{\Delta_*|\emptyset}(\Gamma_\sigma)
	\end{equation} for any square $\Delta \in \B^0_{\Z^2}$ and spin configuration $\sigma \in \{-1,+1\}^{\Z^2}$ equal to $+$ outside $\Delta$. Here, $\emptyset$ denotes the empty contour configuration. Using \eqref{eq:iscr2} and the spin-flip symmetry of the Ising model, one can show that if the gas measure $\mu^P:= \lim_{\Delta \nearrow \Z^2_*} \mu^P_{\Delta_*|\emptyset}$ is sufficiently diluted (which occurs at low temperatures) then the infinite-volume measures $\mu^+$ and $\mu^-$ in the Ising model are distinct. Therefore, by changing the (local) spin variables into new (non-local) contour variables, the proof of the existence of a phase transition reduces to the proof of some form of diluteness of the contour measure. We refer to \cite{SJS} for details.

In our present context, however, this model is of interest also for another reason: it constitutes the canonical example of a (physically relevant) non-Gibbsian system. Indeed, for any $\gamma_x \in \Z^2_* \times G$, $\Lambda \in \B^0_{\Z^2_*}$ and $\eta \in \mathcal{N}(\Z^2_* \times G)$ we have that the energy leap 
	$$
	\Delta E_{\Lambda|\eta} (\gamma_x) = \left\{\begin{array}{ll} +\infty & \text{ if $\gamma_x \cap \langle \eta \rangle \neq \emptyset$ or $\gamma_x \cap (\Z^2_* - \Lambda) \neq \emptyset$}\\ \\ 0 & \text{ otherwise}\end{array}\right.
	$$ depends on the volume $\Lambda$ through the restriction imposed by the interaction term $\Phi^{(1)}_{\Lambda}$, implying that the model is not Gibbsian. However, the limit 
	$$
	\lim_{\Lambda \nearrow \Z^2_*} \Delta E_{\Lambda|\eta}(\gamma_x) = \left\{\begin{array}{ll} +\infty & \text{ if $\gamma_x \cap \langle \eta \rangle \neq \emptyset$}\\ \\ 0 & \text{ otherwise}\end{array}\right.
	$$ exists for all choices of $\eta$ and $\gamma_x$, so that the conditions in Definition \ref{assump} are still satisfied. The restriction imposed by $\Phi^{(1)}_\Lambda$ must be, nonetheless, included for \eqref{eq:iscr2} to hold, as $\mu^I_{\Delta|+}$ is supported on configurations $\sigma$ such that $\Gamma_\sigma$ is always contained in $\Delta_*$. 

\subsubsection{Other examples} More in general, our treatment is also adapted to handle systems in the following general classes:

\bigskip
\paragraph{\it General contour ensembles.}  The Peierls contours discussed above are particularly simple because the Ising Hamiltonian is symmetric under the overall flipping of configurations.  More general non-symmetric cases are the object of study of Pirogov-Sinai theory \cite{PS}.  The main features of contours defined in this theory are the following: (i) Contours are ``thick" subsets formed by collections of plaquettes, (ii) contours include some additional information (color, configurations on both sides or, in general, the configuration on the relevant plaquettes), (iii) each reference configuration has a specific contour ensemble, (iv) contour ensembles are not of physical nature and only external contours coincide with the physical ``defects'' in the presence of reference configurations, (v) contour weights include ratios of partition functions that must be bounded so to obtain exponential expressions similar to the one in \eqref{eq:iscr}.  The last property is encoded in the expression ``contours must satisfy a Peierls condition''.  If these conditions are met, by proceeding as for the Peierls contours one can show that the diluteness of a contour ensemble implies the existence of a measure ``tilted" towards the corresponding boundary configuration. We observe that PS contours ensembles do not fit the Gibbsian framework for the same reason that in the Peierls contours model. Our results applied to these PS contour ensembles can yield not only proofs of the existence of phase transitions in the associated spin systems, but also the stability of the resulting phases with respect to perturbations and discretizations. This will be exploited in a subsequent publication \cite{fersagps}.

\bigskip
\paragraph{\it General polymer models.} These models ---introduced by Gruber and Kunz \cite{grukun71}--- involve general geometrical objects subject to a general hard-core condition defined in terms of a ``compatibility" relation (see \cite{kotpre86,dob96a}).  They are the traditional target of cluster-expansion or closely related methods \cite{bisferpro10}.  Our results extend uniqueness and mixing properties to a larger region of parameters than expansion-based treatments, at the cost of sacrificing analyticity considerations.  

\bigskip \paragraph{\it General point processes}  Point processes are the genesis of the ancestor algorithm exploited in this paper. All the models presented in one of the original publications \cite{FFG2} fall within the scope of our treatment: area-interacting processes, Strauss process, loss networks, random cluster model. In fact, our results apply to models that combine the generality of polymer models ---that do not require a geometric underlying space--- with that of point processes ---that allow soft as well as hard-core interactions.  A full presentation of the ancestor algorithm in such a general framework, and involving even weaker dilution requirements, is the object of a  separate paper \cite{sagdm}.

	\subsection{Approximation families}\label{sec:aprox} We now describe the perturbations of the configuration space that will be considered on the target model. These include, but are not limited to, discretization schemes on both the location space and spin set. 
	
	\begin{defi} \label{discdefi}
		A family $\mathfrak{D}=(D_\varepsilon)_{\varepsilon \geq 0}$ of measurable applications $D_\varepsilon : S \times G \to S\times G$ is called an \textit{approximation family} if the following conditions are satisfied:
		\begin{enumerate}
			\item [i.] For any $B \in \B^0_{S \times G}$ and $\delta > 0$ there exists $B^{(\delta)} \in \B^0_{S \times G}$ with $\bigcup_{0 \leq \varepsilon \leq \delta} D^{-1}_{\varepsilon}(B) \subseteq B^{(\delta)}$.
			\item [ii.] For any $B \in \B^0_{S \times G}$ and $\delta > 0$ there exists $B_{(\delta)} \in \B^0_{S \times G}$ with $\bigcup_{0 \leq \varepsilon \leq \delta} D_{\varepsilon}(B) \subseteq B_{(\delta)}$. 
			\item [iii.] There exists $a: \R_{\geq 0} \rightarrow \R_{\geq 0}$ with $\lim_{\varepsilon \rightarrow 0^+} a(\varepsilon) = a(0) = 0$ such that 
			$$
			d_{S \times G}( D_\varepsilon(\gamma_x), \gamma_x ) \leq a(\varepsilon)
			$$ for every $\gamma_x \in S \times G$ and $\varepsilon \geq 0$.
			
		\end{enumerate} The application $D_\varepsilon$ is called the $\varepsilon$-\textit{approximation operator}. 
	\end{defi} We note that conditions (i)-(ii) are merely technical requirements needed for the proofs, the essence of Definition \ref{discdefi} is contained in (iii). In fact, if the metric structure on $S \times G$ is such that $
	B_\delta = \{ \gamma_x \in S \times G : d(\gamma_x,B) \leq \delta\} \in \B^0_{S \times G}
	$ holds for any $B \in \B^0_{S \times G}$ then conditions (i)-(ii) are immediately satisfied and they can be removed from Definition \ref{discdefi}. On a side note, we observe that by definition $D_0$ is always the identity operator on $S \times G$. This is so for notational convenience.
	
	\begin{ejem} Some natural examples of approximation families include:
		\begin{enumerate}
			\item [$\bullet$] \textit{Spatial translations:} Defined on $S=\R^d$, given for each $\varepsilon > 0$ by 
			$$
			D^{\text{tr}}_\varepsilon(x,\gamma) = (x+\varepsilon \cdot v, \gamma)
			$$ 
			for some fixed unit vector $v \in \R^d$.   
			\item [$\bullet$] \textit{Spatial discretizations:} Defined on $S = \R^d$, given for each $\varepsilon > 0$ by 
			$$
			D^{\text{ds}}_\varepsilon ( x, \gamma ) = ( x_\varepsilon , \gamma )
			$$ where, for $x=(x_1,\dots,x_d) \in \R^d$, we write
			\begin{equation}\label{spatialdiscret}
			x_\varepsilon := \left( \varepsilon \left[\frac{x_1}{\varepsilon}\right],\dots,\varepsilon \left[\frac{x_d}{\varepsilon}\right]\right).
			\end{equation}
			\item [$\bullet$] \textit{Spin rotations:}  Defined on $G = S^{d-1}$, the unit sphere in $d$-dimensions, given for each $\varepsilon > 0$
			$$
			D^{\text{rot}}_\varepsilon ( x, \gamma ) = ( x ,  \varepsilon \cdot R \gamma) 
			$$ where $R$ is some fixed rotation. 
			
			\item [$\bullet$] \textit{Spin discretizations:}  Defined on $G=[0,\pi)$, given for each $\varepsilon > 0$ by 
			\begin{equation}\label{spindiscret}
			D^{\text{sds}}_\varepsilon ( x, \gamma ) = \left( x , \varepsilon \left[\frac{\gamma}{\varepsilon}\right] \right).
			\end{equation}
		\end{enumerate}
	\end{ejem}
	
	Approximations can be composed giving rise, for instance, to operators of the form  $D^{\text{ds}}_{\varepsilon_1}D^{\text{sds}}_{\varepsilon_2}$ that discretize both space and spin and may depend on more than one parameter. To simplify the exposition we will always assume that $\varepsilon \in \R_{\geq 0}$, but we point out that the extension to the case in which $\varepsilon=(\varepsilon_1,\dots,\varepsilon_k)$ is a vector of parameters is straightforward. 
	
	In the sequel, for notational convenience we shall write $\gamma_x^\varepsilon$ instead of $D_\varepsilon (\gamma_x)$. Moreover, for $\xi \in \mathcal{N}(S \times G)$ and each $\varepsilon > 0$ we define $D_\varepsilon(\xi) \in \mathcal{N}(S \times G)$ as the particle configuration given by the standard representation
	$$
	D_\varepsilon(\xi)= \sum_{\gamma_x \in \langle \xi \rangle} m(\gamma_x)\delta_{\gamma_x^\varepsilon}.
	$$ The fact that $D_\varepsilon(\xi)$ is indeed locally finite follows from (i) in Definition \ref{discdefi}. \mbox{Furthermore,} it follows from Lemma \ref{lemaconvaga} below 
	 that $\lim_{\varepsilon \rightarrow 0^+} D_\varepsilon(\xi) = \xi$ vaguely. To simplify notation, we may sometimes write $\xi^\varepsilon$ instead of $D_\varepsilon(\xi)$.

\section{Main results}

The main results featured in this article concern the particular class of \textit{heavily diluted} gas models, which we introduce now. 

\begin{defi} \label{defi:heavy}
A model $(\nu,H)$ satisfying Definition \ref{assump} is said to be \textit{heavily diluted} if there exists a measurable function $q: S \times G \to [1,+\infty)$ such that
	\begin{equation}\label{hdiluted}
	\alpha_q^{\nu,H}:=\sup_{\gamma_x \in S\times G} \left[\frac{e^{-\Delta E}}{q(\gamma_x)} \int_{I(\gamma_x)} q(\tilde{\gamma}_y)d\nu(\tilde{\gamma}_y)\right] < 1.
	\end{equation} The quantity $\alpha_q^{\nu,H}$ is called the $q$-\textit{diluteness coefficient} and $q$ is called the \textit{size function}. 
\end{defi} 

Heavily diluted models have a unique gas measure, as the following result in \cite{SJS} shows. 

\begin{teo}[Uniqueness of the gas measure in heavily diluted models] \label{teounigibbs} $\\$
	Let $(\nu,H)$ be a heavily diluted model on $S \times G$. Then: 
	\begin{enumerate}
		\item [i.] The local limit $\mu:= \lim_{\Lambda \nearrow S} \mu_{\Lambda|\eta}$ exists and coincides for any $\eta \in \mathcal{N}(S \times G)$ with finite $H$-interaction range. In particular, $(\nu,H)$ admits a unique gas measure.
		\item [ii.] If $(\nu,H)$ is Gibbsian then $\mu$ constitutes the unique Gibbs measure of the model.
	\end{enumerate}
\end{teo}

To state our results we need to introduce first the notion of negligible event on $\mathcal{N}(S \times G)$.

\begin{defi}\label{diset} We say that an event $N \subseteq \mathcal{N}(S \times G)$ is \textit{dynamically negligible} for a given intensity measure $\nu$ on $S \times G$ if it satisfies the following properties:
	\begin{enumerate}
		\item [i.] $N$ is a $\sigma$-local $\pi^\nu$-null event.
		\item [ii.] $N$ is closed by addition of particles, i.e. $\eta \in N , \eta \preceq \xi \Longrightarrow \xi \in N$, where $\eta \preceq \xi$ whenever their standard representations satisfy $\langle \eta \rangle \subseteq \langle \xi \rangle$ and $m_\eta(\gamma_x) \leq m_\xi(\gamma_x)$ for every $\gamma_x \in Q_\eta$.
	\end{enumerate}
\end{defi}

Examples of dynamically negligible sets will be given in the applications of Section \ref{secapps}. 
As illustrated in Section \ref{secapps}, most realistic continuum models will not satisfy the limit \eqref{eq:dns} in the hypotheses of Theorem \ref{convdis} below for \textit{every} configuration, but will rather do so only for configurations outside a dynamically negligible set. Thus, for our results to be of any real use in the continuum, it is necessary to allow for violations of \eqref{eq:dns} inside such sets. Fortunately, our results will still hold under this weaker hypothesis. Later in Section \ref{secdis}, we will introduce the even weaker notion of dynamically \textit{impossible} sets when discussing further relaxations to the hypotheses of Theorem \ref{convdis} below.


We are now ready to present our main result. In its statement we add superscripts to the usual notation in order to identify the model which we are referring to.  

\begin{teo}\label{convdis} Let $(D_\varepsilon)_{\varepsilon \geq 0}$ be an approximation family and suppose that $(\nu^\varepsilon,H^\varepsilon)_{\varepsilon \geq 0}$ is a family of gas models such that:
	\begin{enumerate}
		\item [i.] For every $\varepsilon > 0$ the intensity measure $\nu^\varepsilon$ is given by 
		$$
		\nu^\varepsilon := \nu^0 \circ D_\varepsilon^{-1}.
		$$ 
		\item [ii.] There exists a dynamically negligible set $N$ for $\nu^0$ such that
		\begin{equation}\label{eq:dns}
		\lim_{\varepsilon \rightarrow 0^+} \Delta E^{H^\varepsilon}_{\eta^\varepsilon} (\gamma^\varepsilon_x) = \Delta E^{H^0}_\eta (\gamma_x)
        \end{equation} for every $\gamma_x \in S \times G$ and all $\eta \in \mathcal{N}(S \times G)$ with $\eta + \delta_{\gamma_x} \in N^c$. 	
		\item [iii.]
		$$
		\Delta E := \inf_{\varepsilon \geq 0} \left[ \inf_{\Lambda \in \B^0_S} \left[ \inf_{\substack{ \gamma_x \in \Lambda \times G \\ \eta \in \mathcal{N}(S \times G) }}\Delta E^{H^\varepsilon}_{\Lambda|\eta^\varepsilon}(\gamma_x^\varepsilon) \right] \right] > -\infty.
		$$
		\item [iv.] There exists for each $\gamma_x \in S \times G$ a set $V(\gamma_x) \in \B_{S \times G}$ such that:
		\begin{itemize}
		\item[$\bullet$] For every $\varepsilon \geq 0$ one has the inclusion
		\begin{equation} \label{eq:incdis}
		D_\varepsilon^{-1}\left(I^{H^\varepsilon}(\gamma^\varepsilon_x)\right) \subseteq V(\gamma_x).
		\end{equation}
		\item[$\bullet$] There exists a size function $q: S \times G \rightarrow [1,+\infty)$ which verifies 
		\begin{equation} \label{eq:coefdis}
		\alpha_{q}^{\nu^0,V} := \sup_{\gamma_x \in S \times G} \left[ \frac{e^{-\Delta E}}{q(\gamma_x)} \int_{V(\gamma_x)} q(\tilde{\gamma}_y)d\nu^0(\tilde{\gamma}_y)\right] < 1.
		\end{equation}
		\end{itemize}
	\end{enumerate}
	Then:
	\begin{itemize}
	\item[(a)] Each model $(\nu^\varepsilon,H^\varepsilon)$ admits exactly one gas measure $\mu^\varepsilon$.
	\item[(b)]  As $\varepsilon \rightarrow 0^+$, we have the weak convergence
	$$
	\mu^{\varepsilon} \overset{w}{\longrightarrow} \mu^0.
	$$ 
	\item[(c)] There exists a coupling $(\mathcal{Z}^\varepsilon)_{\varepsilon \geq 0}$ of the measures $(\mu^\varepsilon)_{\varepsilon \geq 0}$ such that for any $B \in \B^0_{S \times G}$ there exists (a random) $\varepsilon_B > 0$ verifying that for all $\varepsilon \leq \varepsilon_B$
	\begin{equation}
	\label{eq:couple}
	\mathcal{Z}^\varepsilon_B = D_{\varepsilon}\left( \mathcal{Z}^0_{D^{-1}_\varepsilon(B)}\right).
	\end{equation} 
	In particular, $\mathcal{Z}^\varepsilon \overset{as}{\longrightarrow} \mathcal{Z}^0$ with respect to the vague topology.
	\end{itemize}
\end{teo}

We point out that, in all common situations, the condition $\alpha^{\nu^0,H^0} < 1$ alone is enough to guarantee the validity of (iv) in the statement of Theorem \ref{convdis}. Indeed, as we shall see in Section \ref{secapps}, a set $V(\gamma_x)$ satisfying \eqref{eq:incdis} can generally be obtained by slightly enlarging $I^{H^0}(\gamma_x)$ in some appropriate manner. If $\alpha_q^{\nu^0,H^0}$ is a continuous function of the parameters of the model and all the remaining models $(\nu^\varepsilon,H^\varepsilon)_{\varepsilon > 0}$ are ``sufficiently close'' to $(\nu^0,H^0)$, then performing this slight enlargement will yield a coefficient $\alpha^{\nu^0,V}_q$ very close to $\alpha_q^{\nu^0,H^0}$, so that \eqref{eq:coefdis} holds. 


In principle, Theorem \ref{convdis} deals only with perturbations of the intensity measure which are given by approximations in the sense of Definition \ref{discdefi}. However, one may cover other cases of interest as well by first transferring perturbations in the intensity measure to an effective Hamiltonian prescription and then applying Theorem \ref{convdis}. In this way, we obtain the following important corollary, dealing with absolutely continuous modifications to the intensity measure. This scenario typically represents perturbations in the parameters of the model: fugacity of particles, inverse temperature and interaction range among others.

\begin{cor}\label{convabs} Let $(\nu^\varepsilon,H^\varepsilon)_{\varepsilon \geq 0}$ be a family of diluted models such that:
	\begin{enumerate}
		\item [i.] There exists an intensity measure $\nu$ on $S \times G$ such that $\nu^\varepsilon \ll \nu$ for every $\varepsilon \geq 0$. 
			\item [ii.] There exists a dynamically negligible set $N$ for $\nu$ such that
			$$
			\lim_{\varepsilon \rightarrow 0^+} \Delta \tilde{E}^{\varepsilon}_{\eta} (\gamma_x) = \Delta \tilde{E}^{0}_\eta (\gamma_x),
			$$ for every $\gamma_x \in S \times G$ and all $\eta \in \mathcal{N}(S \times G)$ with $\eta + \delta_{\gamma_x} \in N^c$, where 
			$$
			\Delta \tilde{E}^{\varepsilon}_\eta(\gamma_x) := \lim_{\Lambda \nearrow S} \Delta \tilde{E}^{\varepsilon}_{\Lambda|\eta} (\gamma_x)
			$$ with
			$$
			\Delta \tilde{E}^{\varepsilon}_{\Lambda|\eta} (\gamma_x) := \Delta E^{H^\varepsilon}_{\Lambda|\eta}(\gamma_x) - \log \left( \frac{d\nu^\varepsilon}{d \nu}(\gamma_x)\right).
				$$ 
				\item [iii.]
				$$
				\Delta \tilde{E} := \inf_{\varepsilon \geq 0} \left[ \inf_{\Lambda \in \B^0_S} \left[ \inf_{\substack{ \gamma_x \in \Lambda \times G \\ \eta \in \mathcal{N}(S \times G)}}\Delta \tilde{E}^{\varepsilon}_{\eta}(\gamma_x) \right]\right] > -\infty.
				$$
				\item [iv.] There exists for each $\gamma_x \in S \times G$ a set $V(\gamma_x) \in \B_{S \times G}$ such that:
					\begin{itemize}
						\item[$\bullet$] For every $\varepsilon \geq 0$ one has the inclusion
						\begin{equation} \label{eq:incv}
						I^{H^\varepsilon}(\gamma^\varepsilon_x) \subseteq V(\gamma_x).
						\end{equation}
						\item[$\bullet$] There exists a size function $q: S \times G \rightarrow [1,+\infty)$ which verifies 
						\begin{equation*}
						\alpha_{q}^{\nu,V} := \sup_{\gamma_x \in S \times G} \left[ \frac{e^{-\Delta \tilde{E}}}{q(\gamma_x)} \int_{V(\gamma_x)} q(\tilde{\gamma}_y)d\nu(\tilde{\gamma}_y)\right] < 1.
						\end{equation*}
					\end{itemize}
		\end{enumerate} 
		Then:
		\begin{itemize}
		\item[(a)] Each model $(\nu^\varepsilon,H^\varepsilon)$ admits exactly one gas measure $\mu^\varepsilon$.
		\item[(b)] As $\varepsilon \rightarrow 0^+$, we have the local convergence
	$$
	\mu^{\varepsilon} \overset{loc}{\longrightarrow} \mu^0.
	$$ 	
	    \item[(c)] There exists a coupling $(\mathcal{Z}^\varepsilon)_{\varepsilon \geq 0}$ of the measures $(\mu^\varepsilon)_{\varepsilon \geq 0}$ such that for any $B \in \B^0_{S \times G}$ there exists (a random) $\varepsilon_B > 0$ verifying  $\mathcal{Z}^\varepsilon_{B} = \mathcal{Z}^0_{B}$ for all $\varepsilon \leq \varepsilon_B$. 
	\end{itemize}
	\end{cor}

\section{Applications}\label{secapps}

We discuss here some consequences of our main results for the models in Section \ref{secexamples}. First, we will focus on applications of Corollary \ref{convabs} and then consider other applications using Theorem \ref{convdis} in its full generality. 
The reader should keep in mind that, although we do not state it explicitly in each application, for every local and/or weak convergence of probability measures throughout this section Theorem \ref{convdis} guarantees the existence of a coupling in which the convergence takes place almost surely. 

\subsection{Applications of Corollary \ref{convabs}}

We illustrate the use of Corollary \ref{convabs} by showing the continuity in the parameters of the unique gas measure for the Widom-Rowlinson and Peierls contour models defined in Section \ref{secexamples}.


\subsubsection{The Widom-Rowlison model}

\begin{teo}[The discrete Widom-Rowlinson model]\label{convabswr} Given $k \in \N$ and $\lambda^+_0,\lambda^-_0 > 0$, consider the discrete Widom-Rowlinson model with fugacities $\lambda^{\pm}_0$ and exclusion radius $k$. Then, if 
	\begin{equation}\label{eq:dwrq}
	\alpha_{WR}^{(d)}(\lambda^+_0,\lambda^-_0,k) := \max\{\lambda^-_0,\lambda^+_0\}(2k+1)^d + \min\{\lambda^-_0,\lambda^+_0\}< 1
	\end{equation} there exists an open neighborhood $U$ of $(\lambda^+,\lambda^-)$ such that for any $(\lambda^+,\lambda^-) \in U$ \mbox{the model} with fugacities $\lambda^\pm$ and exclusion radius $k$ admits a unique Gibbs measure $\mu(\lambda^+,\lambda^-)$. Moreover, the application $(\lambda^+,\lambda^-) \mapsto \mu(\lambda^+,\lambda^-)$ is continuous on $U$ in the local topology, i.e. for any $(\lambda^+_*,\lambda^-_*) \in U$ the following local limit holds:
	$$
	\lim_{(\lambda^+,\lambda^-) \rightarrow (\lambda^+_*,\lambda^-_*)} \mu(\lambda^+,\lambda^-) = \mu(\lambda^+_*,\lambda^-_*).
	$$
\end{teo}

\begin{proof} Observe that for any $\gamma_x \in \Z^d \times \{+,-\}$ and $\eta \in \mathcal{N}(\Z^d \times \{+,-\})$ we have 
	\begin{equation}
	\label{eq:idwr}
	I(\gamma_x) = \{ \gamma_x \} \cup \{ \tilde{\gamma}_y : \tilde{\gamma}=-\gamma\,,\,\|x-y\|_\infty \leq k\}.
	\end{equation} and 
	\begin{equation}
	\label{eq:edwr}
	\Delta E_{\eta}(\gamma_x) = \left\{\begin{array}{ll} +\infty & \text{ if $\eta( I(\gamma_x)) > 0$} \\ \\ 0 & \text{ otherwise.}\end{array}\right.
	\end{equation}
	It follows from this that for the model $(\nu^{\lambda^+,\lambda^-},H)$ with fugacities $\lambda^\pm$ and \mbox{exclusion radius $k$} the associated diluteness coefficient in \eqref{hdiluted} for any \textit{constant} size function $q$ is 
	$$
	\alpha^{\nu^{\lambda^+,\lambda^-},H}_q = \sup_{\gamma_x \in \Z^d \times \{+,-\}} e^{-\Delta E} \nu^{\lambda^+,\lambda^-}(I(\gamma_x)) =  \max\{\lambda^-,\lambda^+\}(2k+1)^d + \min\{\lambda^-,\lambda^+\}.
	$$ Therefore, if $\alpha_{WR}^{(d)}(\lambda^+_0,\lambda^-_0,k)< 1$ then there exists an open neighborhood $U$ of $(\lambda^+,\lambda^-)$ such that $\alpha_q^{\nu^{\lambda^+,\lambda^-},H} < 1$ for all $(\lambda^+,\lambda^-) \in U$. In particular, by Theorem \ref{teounigibbs} there exists a unique Gibbs measure $\mu(\lambda^+,\lambda^-)$ of the model $(\nu^{\lambda^+,\lambda^-},H)$ for any $(\lambda^+,\lambda^-) \in U$.  
	
	To see that $(\lambda^+,\lambda^-) \mapsto \mu(\lambda^+,\lambda^-)$ is locally continuous in $U$, we fix $(\lambda^+_\infty,\lambda^-_\infty) \in U$ and check that for any  sequence $(\lambda^+_n,\lambda^-_n)_{n \in \N} \subseteq U$ converging to $(\lambda^+_\infty,\lambda^-_\infty)$ we have 
	$$
	\mu(\lambda^+_n,\lambda^-_n) \overset{loc}{\longrightarrow} \mu(\lambda^+_\infty,\lambda^-_\infty).
	$$ For this it suffices to see that $(\nu^{\lambda^+_n,\lambda^-_n},H)_{n \in \N \cup \{\infty\}}$ satisfies the hypotheses \mbox{of Corollary \ref{convabs}.}
	But notice that if for $n \in \N \cup \{\infty\}$ we write $\nu^n:=\nu^{\lambda^+_n,\lambda^-_n}$ then $\nu^n \ll \nu^\infty$ with density given by
	$$
	\frac{d \nu^n}{d \nu^\infty}(\gamma_x) = \frac{\lambda^+_n}{\lambda^+_\infty} \mathbbm{1}_{\{\gamma = +\}} + \frac{\lambda^-_n}{\lambda^-_\infty} \mathbbm{1}_{\{\gamma = -\}}.
	$$ In particular, we have that for every $\gamma_x \in \Z^d \times \{+,-\}$ and $\eta \in \mathcal{N}(\Z^d \times \{+,-\})$ 
	$$
	\Delta \tilde{E}^n_\eta (\gamma_x) = \Delta E^H_\eta (\gamma_x) - \log \left( \frac{d\nu^n}{d\nu^\infty}(\gamma_x)\right) \longrightarrow \Delta E^H_\eta (\gamma_x)  = \Delta \tilde{E}^\infty_\eta(\gamma_x)
	$$ and also that
	$$
	\Delta \tilde{E} = - \log\left(\sup_{n \in \N} \left[ \max\left\{ \lambda^+_n - \lambda^+_\infty,\lambda^-_n - \lambda^-_\infty\right\}\right]\right) > -\infty
	$$	if $(\lambda^+_n,\lambda^-_n)_{n \in \N}$ is sufficiently close to $(\lambda^+_\infty,\lambda^-_\infty)$. Furthermore, if for each $\gamma_x \in \Z^d \times \{+,-\}$ we choose $V(\gamma_x) := I^H(\gamma_x)$ then the inclusion \eqref{eq:incv} immediately holds for all $n \in \N \cup \{\infty\}$ and 
	$$
	\alpha^{\nu^\infty,V}_q = e^{-\Delta \tilde{E}} \alpha^{\nu^\infty,H}_q < 1
	$$ provided that $(\lambda^+_n,\lambda^-_n)_{n \in \N}$ is sufficiently close to $(\lambda^+_\infty,\lambda^-_\infty)$ so that $\Delta \tilde{E}$ \mbox{is close enough to $0$.} Hence, we see that the sequence $(\nu^{\lambda^+_n,\lambda^-_n},H)_{n \in \N \cup \{\infty\}}$ verifies the hypotheses of \mbox{Corollary \ref{convabs}} (with the limit $n \rightarrow \infty$ replacing the usual $\varepsilon \rightarrow 0$) and this concludes the proof.
\end{proof}

\begin{obs} If one uses the alternative representation of the Widom-Rowlinson model given by the pair $(\tilde{\nu},\tilde{H})$ in Section \ref{sec:dwr} then for any constant size function $q$ one obtains the larger diluteness coefficient
	$$
	\alpha^{\tilde{\nu},\tilde{H}}_q = \max\{\lambda^-,\lambda^+\} ( (2k+1)^d + 1)
	$$ which leads to a smaller uniqueness condition than the one in Theorem \ref{convabswr}. This is why we chose to include the fugacities in the intensity instead of the Hamiltonian prescription. 
\end{obs}

To illustrate the need to consider dynamically impossible events, let us treat the case of the continuum Widom-Rowlinson model. We have the following analogue \mbox{of Theorem \ref{convabswr}.} 

\begin{teo} [The continuum Widom-Rowlinson model]\label{convabscwr} Given $r_0 > 0$ and $\lambda^+_0,\lambda^-_0 > 0$, consider the continuum Widom-Rowlinson model of fugacities $\lambda^{\pm}_0$ and exclusion radius $r_0$. Then, if 
	\begin{equation}
	\label{eq:cwrq}
	\alpha_{WR}^{(c)}(\lambda^+_0,\lambda^-_0,r_0) := \max\{\lambda^-_0,\lambda^+_0\}(2r_0)^d < 1
	\end{equation} there exists an open neighborhood $U$ of $(\lambda^+_0,\lambda^-_0,r_0)$ such that for every $(\lambda^+,\lambda^-,r) \in U$ \mbox{the model} of fugacities $\lambda^\pm$ and exclusion radius $r$ has a unique Gibbs measure $\mu(\lambda^+,\lambda^-,r)$. Moreover, the application $(\lambda^+,\lambda^-,r) \mapsto \mu(\lambda^+,\lambda^-,r)$ is locally continuous on $U$.
\end{teo}

\begin{proof} The first assertion follows as in the proof of Theorem \ref{convabswr} by noticing that now \eqref{eq:idwr} is replaced with 
	$$
	I(\gamma_x) = \{ \tilde{\gamma}_y : \tilde{\gamma}=-\gamma\,,\,\|x-y\|_\infty \leq r\}.
	$$ To see the local continuity, observe that if $(\lambda^+_n,\lambda^-_n,r_n)_{n \in \N}$ converges to $(\lambda^+_\infty,\lambda^-_\infty,r_\infty) \in \R_{> 0}^3$ and for each $n \in \N \cup \{\infty\}$ we denote by $(\nu^n,H^n)$ the model with parameters $(\lambda^+_n,\lambda^-_n,r_n)$ then the convergence
	$$
	\Delta \tilde{E}^n_\eta (\gamma_x) = \Delta E^{H^n}_\eta (\gamma_x) - \log \left( \frac{d\nu^n}{d\nu^\infty}(\gamma_x)\right) \longrightarrow \Delta E^{H^\infty}_\eta (\gamma_x)  = \Delta \tilde{E}^\infty_\eta(\gamma_x)
	$$ may not hold if $\eta + \delta_{\gamma_x}$ is inside the set
	$$
	N = \{ \xi \in \mathcal{N}(\R^d \times \{+,-\}) : \exists\,\, \tilde{\gamma}_y \neq \hat{\gamma}_z \in \langle \xi \rangle \text{ such that } \tilde{\gamma}=-\hat{\gamma} \text{ and }\|y-z\|_\infty = r_0\}.
	$$ Indeed, the problem lies when $\gamma_x$ is at a distance $r_0$ from a particle in $\eta$ of opposite type. If this is the case and $r_n \nearrow r_\infty$ then we can have $\Delta E^{H^n}_\eta (\gamma_x) = 0 \nrightarrow +\infty = \Delta E^{H^\infty}_\eta (\gamma_x)$. However, since $N$ is a dynamically negligible event for $\nu^\infty$, one may disregard these cases and still proceed as in the proof of Theorem \ref{convabswr} to conclude the result. 
\end{proof} 

Similarly, we can obtain an analogous result for the model with generalized interactions. For simplicity, we only state it for its continuum version and omit the details of the proof.

\begin{teo}[The continuum Widom-Rowlinson model with generalized interactions]\label{convabscwrgi} For each $\varepsilon \geq 0$ let us consider the continuum Widom-Rowlinson model with fugacities $\lambda^{\pm}_\varepsilon$ and repulsion vector $(h^\varepsilon,j^\varepsilon_-,j^\varepsilon_+)$. Assume that the following conditions hold:
	\begin{enumerate}
		\item [i.] $\lim_{\varepsilon \rightarrow 0^+} \lambda^+_\varepsilon = \lambda^+_0$ and $\lim_{\varepsilon \rightarrow 0^+} \lambda^-_\varepsilon = \lambda^-_0$.
		\item [ii.] $\lim_{\varepsilon \rightarrow 0^+} h^\varepsilon(r) = h^0(r)$ and $\lim_{\varepsilon \rightarrow 0^+} j^\varepsilon_\pm(r) = j^0_\pm(r)$ for $\mathcal{L}$-almost every $r \geq 0$.
		\item [iii.] $\lim_{\varepsilon \rightarrow 0^+} m_{h^\varepsilon} = m_{h^0}$ and $\lim_{\varepsilon \rightarrow 0^+} m_{j^\varepsilon_\pm} = m_{j^0_\pm}$, where
		$$
		m_{h^\varepsilon}= \sup \{ r \geq 0 : h^\varepsilon(r)\neq 0\}\hspace{1cm}\text{ and }\hspace{1cm} m_{j^\varepsilon_\pm}= \sup \{ r \geq 0 : j^\varepsilon_\pm(r)\neq 0\}.
		$$
		\item [iv.] $\alpha_{WR}(\lambda^{\pm}_0, h^0,j^0_-,j^0_+) < 1$, where for any $\lambda^+,\lambda^- > 0$ and repulsion vector $(h,j_-,j_+)$ we define
		$$
		\alpha_{WR}(\lambda^{\pm}, h,j) := 2^d \max \{ \lambda^- m^d_{j_-} + \lambda^+ \max\{ m_h^d , m^d_{j_+} \} , \lambda^+ m^d_{j_+} + \lambda^- \max\{ m_h^d , m^d_{j_-} \} \}.
		$$                
	\end{enumerate} Then for $\varepsilon \geq 0$ sufficiently small there exists a unique Gibbs measure $\mu^\varepsilon$ of the associated continuum Widom-Rowlinson model. Furthermore, we have the convergence $\mu^\varepsilon \overset{loc}{\rightarrow} \mu^0$.
\end{teo} 

The uniqueness regions of parameters prescribed by the above bounds are comparable ---although a little weaker--- to those obtained via disagreement percolation arguments (see, for instance, \cite{GHM}).  It is not clear, however, whether the latter method can lead to the stability results brought by our approach.  

\subsubsection{The Peierls contour model}

\begin{teo}[The Peierls contours model]\label{convabspcm} Define the coefficient
	$$
	\beta_P = \inf\{ \beta > 0 : \alpha_P(\beta) < 1\}, 
	$$ where 
	\begin{equation}
	\label{eq:pcmq}
	\alpha_{P}(\beta) := \sup_{\gamma_x \in \Z^2_* \times G} \left[\frac{1}{|\gamma_x|} \sum_{\tilde{\gamma}_y : \tilde{\gamma}_y \cap \gamma_x \neq \emptyset } |\tilde{\gamma}_y|e^{-2\beta|\tilde{\gamma}_y|}\right].
	\end{equation} Then, if $\beta > \beta_P$ the model at inverse temperature $\beta$ admits a unique gas measure $\mu(\beta)$. Moreover, the application $\beta \mapsto \mu(\beta)$ is locally continuous on $(0,\beta_P)$.
\end{teo}

\begin{proof} Notice that the coefficient $\alpha_P(\beta)$ is exactly the diluteness coefficient for the model at inverse temperature $\beta$ associated to the size function $q(\gamma_x):=|\gamma_x|$. Hence, uniqueness for $\beta > \beta_P$ follows from Theorem \ref{teounigibbs}. For the continuity, we proceed as in the proof of Theorem \ref{convabswr}. Thus, we fix $\beta_\infty > \beta_P$ and a sequence $(\beta_n)_{n \in \N} \subseteq \R_{> \beta_P}$ converging to $\beta_\infty$ and show that the family  $(\nu^{\beta_n},H)_{n \in \N \cup \{\infty\}}$ satisfies all the hypotheses of Corollary \ref{convabs}. This can be done as in the proof of Theorem \ref{convabswr}, but one has to be careful that  
	$$
	\inf_{n \in \N} \left[\inf_{\gamma_x \in S \times G} \left[-\log\left(\frac{d \nu^{\beta_n}}{d \nu^{\beta_\infty}} (\gamma_x)\right)\right]\right] = \inf_{n \in \N} \left[ \inf_{\gamma_x \in S \times G} 2(\beta_n - \beta_\infty)|\gamma_x|\right] = -\infty
	$$ if $\beta_n < \beta_\infty$ for some $n \in \N$, so that if one takes $\nu$ in the statement of Corollary \ref{convabs} as $\nu^{\beta_\infty}$ then the same argument that in the proof of Theorem \ref{convabswr} does not go through this time. To solve this, we take $\nu:=\nu^{\beta_*}$ with $\beta_*:=\inf_{n \in \N \cup \{\infty\}} \beta_n$ and then proceed as before. 
\end{proof}

We point out that the analogous result also holds for the model in higher dimensions. The condition in Theorem \ref{convabspcm} for the $d$-dimensional model ($d \geq 2$) is satisfied for all $\beta > 0$ such that 
\begin{equation}
\label{eq:cotapeierls}
\sum_{\ell\ge 2d} \ell \,N_\ell \,e^{-2\beta\ell}<1
\end{equation}
where $N_\ell$ denotes the number of contours of perimeter $\ell$. For the two-dimensional model, this condition coincides with the Peierls condition presented in usual textbooks treatments of the Peierls argument. In higher dimensions, however, the Peierls condition is weaker, basically corresponding to changing the factor $\ell$ in the left-hand side by $\ell^{1/(d-1)}$.  Also, additional arguments based on the Borel-Cantelli lemma show that it suffices to have ``$+\infty$'' instead of ``$1$'' in the right-hand side of \eqref{eq:cotapeierls} in order to guarantee a phase transition. \mbox{Of course,} our theorem implies further properties that cannot be directly obtained from the original Peierls argument.

\subsection{Applications of Theorem \ref{convdis}} 

We now discuss some further applications which use Theorem \ref{convdis} in its full generality. We shall focus in discretization schemes applied to the thin-rods and Widom-Rowlinson models and determine diluteness regimes in which the Gibbs measure of the continuum model is the scaling limit of the equilibrium measure in the discretized model.

\subsubsection{The Widom-Rowlinson model and its ``universality'' class}

\begin{teo}\label{diswr} For $\lambda_0,r_0 > 0$ such that $\lambda_0 (2r_0)^d < 1$ we have the following:
	\begin{enumerate}
		\item [i.] The continuum Widom-Rowlinson model with fugacity $\lambda_0$ and exclusion radius $r_0$ in $\R^d$ admits exactly one Gibbs measure, which we shall denote by $\mu^0$.
		\item [ii.] The discrete Widom-Rowlinson model with fugacity $\varepsilon^d\lambda_0$
		and exclusion radius $\frac{r_0}{\varepsilon}$ in $\Z^d$ admits exactly one Gibbs measure $\tilde{\mu}^\varepsilon$ if $ 0 < \varepsilon < (2^dr_0)^{-\frac{1}{d}} - r_0$.
		\item [iii.] Provided that $0 < \varepsilon < (2^dr_0)^{-\frac{1}{d}} - r_0$, as $\varepsilon \rightarrow 0$ we have
		$$
		\tilde{\mu}^\varepsilon \circ i_\varepsilon^{-1} \overset{d}{\longrightarrow} \mu^0
		$$ where for each $\varepsilon > 0$ we define the \textit{shrinking map} $i_\varepsilon : \Z^d \times \{+,-\} \to \R^d \times \{+,-\}$ by the formula
		$$
		i_\varepsilon ( x , \gamma ) := (\varepsilon\cdot x, \gamma).
		$$
	\end{enumerate}
\end{teo}

\begin{proof} Assertions (i)-(ii) follow at once from Theorem \ref{teounigibbs} and \eqref{eq:cwrq}-\eqref{eq:dwrq}, respectively. To show (iii), we consider the spatial discretization family $(D^{\text{ds}}_\varepsilon)_{\varepsilon \geq 0}$ given by \eqref{spatialdiscret} and for each $\varepsilon \geq 0$ set:
	\begin{enumerate}
		\item [$\bullet$] The intensity measure $\nu^\varepsilon$ as $\nu^\varepsilon := \nu^0 \circ (D^{\text{ds}}_\varepsilon)^{-1}$, where $\nu := \lambda_0 \mathcal{L}^d \times (\delta_+ + \delta_-).$
		\item [$\bullet$] The Hamiltonian prescription $H^\varepsilon$ as the one given by the $2$-body interaction 
		\begin{equation}\label{eq:dwr3}
		\Phi^{(2)}_\varepsilon(\gamma_x,\tilde{\gamma}_y):= \left\{ \begin{array}{ll} +\infty &\text{if $x=y$ and $\varepsilon > 0$}\\ \\ +\infty & \text{if }0< \|x-y\|_\infty \leq r_0 \text{ and }\gamma \neq \tilde{\gamma}\\ \\ 0 &\text{otherwise.}\end{array}\right.
		\end{equation} 
	\end{enumerate} Notice that $(\nu^0,H^0)$ is precisely the continuum Widom-Rowlinson model of fugacity $\lambda_0$ and exclusion radius $r_0$, whereas for each $\varepsilon > 0$ the pair $(\nu^\varepsilon,H^\varepsilon)$ constitutes essentially the $i_\varepsilon$-shrunken discrete model of fugacity $\varepsilon^d \lambda_0$ and exclusion radius $\frac{r_0}{\varepsilon}$. More precisely, for every $\Lambda \in \B^0_{\Z^d}$ and $\varepsilon > 0$ we have
	\begin{equation}\label{igualdaddis}
	\mu^\varepsilon_{i_\varepsilon(\Lambda)|\emptyset} = \tilde{\mu}^\varepsilon_{\Lambda|\emptyset} \circ i_\varepsilon^{-1}
	\end{equation} where $\tilde{\mu}^\varepsilon_{\Lambda|\emptyset}$ is the Boltzmann-Gibbs distribution with empty boundary condition associated to the discrete Widom-Rowlinson model and $\mu^\varepsilon_{\Lambda|\emptyset}$ is the one associated to $(\nu^\varepsilon,H^\varepsilon)$. Thus, by taking the limit as $\Lambda \nearrow \Z^d$, Theorem \ref{teounigibbs} yields for $0 < \varepsilon < (2^dr_0) ^{-\frac{1}{d}} - r_0 $
	$$
	\mu^\varepsilon = \tilde{\mu}^\varepsilon \circ i_{\varepsilon}^{-1}
	$$ where $\mu^\varepsilon$ is the unique Gibbs measure of the gas model given by the pair $(\nu^\varepsilon, H^\varepsilon)$. Hence, to conclude the desired convergence it suffices to show that the family $(\nu^\varepsilon,H^\varepsilon)_{\varepsilon \geq 0}$ is under the hypotheses of Theorem \ref{convdis}. But notice that:
	\begin{enumerate}
		\item [$\bullet$] (i) in the hypotheses holds trivially by the choice of measures $\nu^\varepsilon$.
		\item [$\bullet$] $\lim_{\varepsilon \rightarrow 0^+} \Delta E^{\varepsilon}_{\eta^\varepsilon} (\gamma_x^\varepsilon) = \Delta E^{0}_\eta (\gamma_x)$ for all $\eta \in \mathcal{N}( \R^d \times \{+,-\})$ and $\gamma_x \in  \R^d \times \{+,-\}$ such that $\eta + \delta_{\gamma_x}$ is outside the dynamically negligible set $N_1 \cup N_2$ for $\nu^0$, where
		\begin{equation}\label{eq:dimpwr}
		N_1 = \{ \xi \in \mathcal{N}(\R^d \times \{+,-\}) : \sigma( \{y\} \times \{+,-\} ) > 1 \text{ for some }y \in \R^d \}
		\end{equation} and 
		$$
		N_2= \{ \xi \in \mathcal{N}(\R^d \times \{+,-\}) : \exists\,\, \tilde{\gamma}_y, \hat{\gamma}_z \in \langle \xi \rangle \text{ with } \tilde{\gamma}=-\hat{\gamma} \text{ and }\|y-z\|_\infty = r_0\}.
		$$	
		\item [$\bullet$] $\Delta E :=\inf_{\varepsilon \geq 0} \Delta E^{H^\varepsilon}=0$ by the repulsive nature of the interaction.
		\item [$\bullet$] If for $\gamma_x \in \R^d \times \{+,-\}$ we take 
		$$
		V(\gamma_x) := \{(y,\gamma) : \|y-x\|_\infty < \delta\} \cup \{(z,-\gamma) : \|z-x\|_\infty \leq r_0+\delta\}
		$$ with $\delta > 0$ sufficiently small so that for any constant size function $q$ we have 
		$$
		\alpha_q^{\nu^0,V} = \sup_{\gamma_x \in \R^d \times \{+,-\}} \nu^0(V(\gamma_x)) = \lambda_0( (2(r_0+\delta))^d + (2\delta)^d)< 1,
		$$ then $I^{H^\varepsilon}(\gamma^\varepsilon_x) \subseteq V(\gamma_x)$ for every $\varepsilon \geq 0$ sufficiently small.
	\end{enumerate}
	Thus, the hypotheses of Theorem \ref{convdis} are satisfied and from this the result follows. 
\end{proof}

Theorem \ref{diswr} in particular shows that, for the case of the Widom-Rowlinson model, in order to simulate the Gibbs measure of the continuum system the method of sampling from the Gibbs measure of the discrete model obtained by discretizing both the interactions and the configuration space indeed yields a faithful approximation of the desired distribution, at least whenever in the heavily diluted regime and the discretization is sufficiently refined. It is clear that Theorem \ref{convdis} guarantees that the same statement holds in general for other types of gas models as well. We point out that, in general, this procedure is not equivalent to directly discretizing the continuum measure. Indeed, the proof of Theorem \ref{convdis} shows that \eqref{eq:couple} is in general as much as one can expect, although this is not enough to imply that for any fixed $\varepsilon > 0$ the Gibbs measure of the discrete system $\mu^\varepsilon$ coincides \mbox{with $D_\varepsilon(\mu^0)$,} the discretization of the continuum Gibbs measure. As a matter of fact, the latter may sometimes fail to be a Gibbsian at all, i.e. there exists no Gibbsian gas model for which $D_\varepsilon(\mu^0)$ is a Gibbs measure (see e.g. [16]).

Another interesting interpretation the can be made of Theorem \ref{diswr} is that it portrays the continuum Widom-Rowlinson model as a ``universality class'' to which many families of discrete models converge when properly rescaled. Indeed, the only relevant information about the discrete models used in the proof was that the additional interaction forbidding multiple particles per site vanishes as $\varepsilon \rightarrow 0$, which suggests that the same convergence should also hold for other discrete models with extra interactions which become negligible in the continuum limit. At least for heavily diluted gas models, this is indeed the case and is essentially a consequence of Theorem \ref{convdis}. We illustrate this fact below by considering the particular case of the nematic thin rods model as an example. 

\begin{teo}
	For $\lambda_0,r_0 > 0$ such that $4\lambda_0 r_0^2 < 1$ we have the following:
	\begin{enumerate}
		\item [i.] The continuum Widom-Rowlinson model with fugacity $\lambda_0$ and exclusion radius $r_0$ in $\R^2$ admits exactly one Gibbs measure, which we shall denote by $\mu^0$.
		\item [ii.] The nematic thin rods model in $\Z^2$ with fugacity $\varepsilon^2\lambda_0$
		and rod length $2\frac{r_0}{\varepsilon}$ admits exactly one Gibbs measure $\tilde{\mu}^\varepsilon$ if $0 < \varepsilon < \sqrt{2r_0^2 +\frac{1}{2\lambda_0}} - 2r_0$.
		\item [iii.] Provided that $0 < \varepsilon < \sqrt{2r_0^2 +\frac{1}{2\lambda_0}} - 2r_0$, as $\varepsilon \rightarrow 0$ we have
		$$
		\tilde{\mu}^\varepsilon \circ i_\varepsilon^{-1} \overset{d}{\longrightarrow} \mu^0.
		$$ 
	\end{enumerate}
\end{teo}

\begin{proof} The proof is almost identical to that of Theorem \ref{diswr}. One only has to replace \eqref{eq:dwr3} with 
	$$
	\Phi^{(2)}_\varepsilon(\gamma_x,\tilde{\gamma}_y):= \left\{ \begin{array}{ll} +\infty & \text{ if $\varepsilon > 0$, $\gamma=\tilde{\gamma}=+$ and $|x_1-y_1|\leq r_0$}\\ \\ +\infty & \text{ if $\varepsilon > 0$, $\gamma=\tilde{\gamma}=-$ and $|x_2-y_2|\leq r_0$}\\ \\ +\infty & \text{if }0< \|x-y\|_\infty \leq r_0 \text{ and }\gamma \neq \tilde{\gamma}\\ \\ 0 &\text{otherwise.}\end{array}\right.
	$$ and \eqref{eq:dimpwr} with $N_+ \cup N_-$, where 
	$$
	N_+ = \{ \xi \in \mathcal{N}(\R^2 \times \{+,-\}) : \exists\,\,\tilde{\gamma}_y\neq\hat{\gamma}_z \in \langle \xi \rangle \text{ with $\tilde{\gamma}=\hat{\gamma}=+$ and $y_1=z_1$}\}
	$$ and 
	$$
	N_- = \{ \xi \in \mathcal{N}(\R^2 \times \{+,-\}) : \exists\,\,\tilde{\gamma}_y\neq\hat{\gamma}_z \in \langle \xi \rangle \text{ with $\tilde{\gamma}=\hat{\gamma}=-$ and $y_2=z_2$}\}.
	$$ The rest of the proof goes through exactly as before. We leave the details to the reader.
\end{proof}

\subsubsection{The thin-rods model}

For simplicity, we focus only on the continuum model in $\R^2$. We begin by establishing the continuity in the parameters of the model with a fixed finite number of orientations. In the sequel, given a number $k \in \N$, a vector of $k$ orientations $\vec{\theta}=(\theta_1,\dots,\theta_k) \in [0,\pi)^k$ and a probability vector $\vec{p}=(p_1,\dots,p_k) \in \R^k$, we shall write $\rho^{\vec{\theta},\vec{p}}$ for the orientation measure given by 
$$
\rho^{\vec{\theta},\vec{p}} := \sum_{i=1}^k p_i \delta_{\theta_i}.
$$

\begin{teo}[The thin rods model in $\R^2$]\label{convdistrm} Fix $k \in \N$ and, given $\lambda_0,l_0 > 0$, $\vec{\theta}_0 \in [0,\pi)^k$ and a probability vector $\vec{p}_0 \in \R^k$, consider the thin rods model in $\R^2$ with fugacity $\lambda_0 > 0$, rod length $2l_0 > 0$ and orientation measure $\rho^{\vec{\theta}^0,\vec{p}^0}$. If  
	\begin{equation}
	\label{eq:trmq}
	\alpha_{TR}^{(c)}(\lambda_0,l_0,\rho^{\vec{\theta}_0,\vec{p}_0}) := 4\lambda^0 (l^0)^2 \sup_{\gamma \in [0,\pi)} \left[ \int_0^\pi|\sin(\gamma - \tilde{\gamma})|d\rho^{\vec{\theta}_0,\vec{p}_0}(\tilde{\gamma})\right] < 1
	\end{equation} then there exists an open neighborhood $U$ of $(\lambda_0,l_0,\vec{\theta}_0,\vec{p}_0)$ such that for all $(\lambda,l,\vec{\theta},\vec{p}) \in U$ the model with parameters $(\lambda,l,\vec{\theta},\vec{p})$ has a unique Gibbs measure $\mu(\lambda,l,\vec{\theta},\vec{p})$. \mbox{Moreover,} the applications $(\lambda,l,\vec{\theta},\vec{p}) \mapsto \mu(\lambda,l,\vec{\theta},\vec{p})$ and $(\lambda,l,\vec{p}) \mapsto \mu(\lambda,l,\vec{\theta}_0,\vec{p})$ are respectively \textit{weakly} continuous on $U$ and \textit{locally} continuous on $U_{\vec{\theta}_0} = \{(\lambda,l,\vec{p}) : (\lambda,l,\vec{\theta}_0,\vec{p}) \in U\}$.
\end{teo}

\begin{proof} A straightforward calculation shows that if $(\nu^0,H^0)$ is the model with parameters $(\lambda_0,l_0,\vec{\theta}_0,\vec{p}_0)$ then the diluteness coefficient associated to any constant size function $q$ is 
	$$
	\alpha^{\nu^0,H^0}_q = 4\lambda^0 (l^0)^2 \sup_{\gamma \in [0,\pi)} \left[ \int_0^\pi|\sin(\gamma - \tilde{\gamma})|d\rho^{\vec{\theta}_0,\vec{p}_0}(\tilde{\gamma})\right].
	$$ From this the first assertion about uniqueness of the Gibbs measure immediately follows. To establish the weak continuity, we fix $(\lambda_\infty,l_\infty,\vec{\theta}_\infty,\vec{p}_\infty) \in U$ and an arbitrary sequence $(\lambda_n,l_n,\vec{\theta}_n,\vec{p}_n)_{n \in \N} \subseteq U$ converging to $(\lambda_\infty,l_\infty,\vec{\theta}_\infty,\vec{p}_\infty) \in U$ and show that 
	$$
	\mu(\lambda_n,l_n,\vec{\theta}_n,\vec{p}_n) \overset{d}{\longrightarrow} \mu(\lambda_\infty,l_\infty,\vec{\theta}_\infty,\vec{p}_\infty).
	$$ For this purpose, we define the approximation family $\mathcal{D}=(D_n)_{n \in \N \cup \{\infty\}}$ by the formula
	$$
	D_n(x,\gamma) =\left\{\begin{array}{ll} (x,\vec{\theta}_n(i)) & \text{ if $\gamma=\vec{\theta}_\infty(i)$ for some $i=1,\dots,k$} \\ \\ (x,\gamma) & \text{ otherwise,}\end{array}\right.
	$$ where for each $n \in \N \cup \{\infty\}$ and $i=1,\dots,k$, $\vec{\theta}_n(i)$ denotes the $i$-th coordinate of $\vec{\theta}_n(i)$. Furthermore, for each $n \in \N \cup \{\infty\}$ consider the Hamiltonian prescription $H^n$ given by the potential $\Phi_n = (\Phi^{(1)}_n,\Phi^{(2)}_n)$, where
	$$
	\Phi^{(1)}_n(\gamma_x) = \left\{\begin{array}{ll} - \log\left( \frac{\lambda_n \vec{p}_n(i)}{\lambda_\infty \vec{p}_\infty(i)}\right) & \text{ if $\gamma=\vec{\theta}_n(i)$ for some $i=1,\dots,k$} \\ \\ 0 & \text{ otherwise,}\end{array}\right.
	$$ and
	$$
	\Phi^{(2)}(\gamma_x,\tilde{\gamma}_y) := \left\{ \begin{array}{ll} +\infty &\text{if }(L^{l_n}_\gamma + x) \cap (L^{l_n}_{\tilde{\gamma}} + y) \neq \emptyset\\ 0 &\text{otherwise.}\end{array}\right.
	$$ Now, if we write $\nu^n:=\nu^\infty \circ D_n^{-1}$ where $\nu^\infty:= \lambda_\infty \mathcal{L}^2 \times \rho^{\vec{\theta}_\infty,\vec{p}_\infty}$, then it is not hard to see that the model $(\nu^n,H^n)$ is equivalent to the thin rods model with \mbox{parameters $(\lambda_n,l_n,\vec{\theta}_n,\vec{p}_n)$.} The result now follows at once upon noticing that the sequence $(\nu^n,H^n)_{n \in \N \cup \{\infty\}}$ satisfies the hypotheses of Theorem \ref{convdis} if $(\lambda_n,l_n,\vec{\theta}_n,\vec{p}_n)_{n \in \N}$ is sufficiently close to $(\lambda_\infty,l_\infty,\vec{\theta}_\infty,\vec{p}_\infty)$. Finally, to obtain the local continuity whenever the vector of orientations is fixed we follow a similar argument using Corollary \ref{convabs} instead.
\end{proof}

To conclude, we consider the case of a continuum of possible orientations and derive the analogue of Theorem \ref{diswr} in this context for spin discretizations. For definiteness, \mbox{we state the result} for the continuum model with an uniform orientation measure on $[0,\pi)$. 

\begin{teo} Given $\lambda, l > 0$, let us consider the thin rods model on $\R^2$ with fugacity $\lambda$, rod length $2l$ and uniform orientation measure $\rho_u$ on $[0,\pi)$. Then, if 
	$$
	\alpha^{(c)}_{TR}(\lambda,l,\rho_u)= \frac{8}{\pi}\lambda l^2 < 1
	$$ there exists a unique Gibbs measure $\mu$ of the model and for every $n \in \N$ sufficiently large the model with orientation measure 
	$$
	\rho^n_u := \frac{1}{n} \sum_{i=0}^{n-1} \delta_{\frac{i}{n}}
	$$ also has a unique Gibbs measure $\mu^n$ which, furthermore, satisfies $\mu^n \overset{d}{\rightarrow} \mu$.
\end{teo}

We omit the proof of this result since it goes very much along the lines of \mbox{Theorem \ref{diswr}} but using the spin discretization family introduced in \eqref{spindiscret} instead.

\section{The Fern\'andez-Ferrari-Garcia dynamics}\label{secffg}

In this section we study the Fern\'andez-Ferrari-Garcia dynamics first \mbox{introduced in \cite{FFG1}}. In their work the authors focus on the Peierls contours model of Section \ref{secexamples} and show that, for a sufficiently large value of the inverse temperature $\beta$, the unique gas measure of this contour model can be realized as the unique invariant measure of this dynamics. In \cite{SJS} it is shown that this result can be extended to the broader class of gas models, where the sufficiently large inverse temperature condition is replaced by the one in \eqref{hdiluted}. This extension of the dynamics will be the main tool used in the proof of Theorem \ref{convdis}. We give a brief overview of it now.

Given a gas model $(\nu,H)$, we consider the following dynamics on $\mathcal{N}(S\times G)$:
\begin{enumerate}
    \item [$\bullet$] At rate $e^{-\Delta E}$ the birth of new animals is proposed with intensity given by $\nu$.
	\item [$\bullet$] Each $\gamma_x$ proposed for birth will be effectively born with probability $e^{-(\Delta E_{\eta}(\gamma_x)-\Delta E)}$, where $\eta$ is the state of the system at the time in which the birth of $\gamma_x$ is proposed.
	\item [$\bullet$] Every animal which has effectively been born will have an independent lifetime, with exponential distribution of mean 1.
	\item [$\bullet$] After its lifetime has expired, each animal dies and vanishes from the configuration.
\end{enumerate} 
Thus, the infinitesimal generator $L$ associated to such dynamics is formally given by 
\begin{equation}
\label{generator}
L(f)(\sigma) = \sum_{\gamma_x \in \langle \sigma \rangle}\sigma(\gamma_x)(f(\sigma - \delta_{\gamma_x}) - f(\gamma_x)) + \int e^{- \Delta E_\sigma (\gamma_x)}(f(\sigma + \delta_{\gamma_x}) - f(\gamma_x))d\nu(\gamma_x)
\end{equation} for any bounded local function $f:\mathcal{N}(S \times G) \to \R$. The following result is shown in \cite{SJS}. 

\begin{teo}
	\label{teoffg}
	If $(\nu,H)$ is heavily diluted then the Markov process with generator $L$ exists and has the unique gas measure of $(\nu,H)$ as its unique invariant measure. Furthermore, there exists a time-stationary construction $\mathcal{K}=(\mathcal{K}_t)_{t \in \R}$ of the process.
\end{teo} 

For our purposes it will be convenient to understand the basics of the proof \mbox{of this result,} so we summarize it below. The idea is to construct $\mathcal{K}$ explicitly as a suitable thinning of a non-interacting birth-and-death process on $S \times G$ with the appropriate rates. 

To make matters more precise, let us consider the product space $\mathcal{C}= (S \times G) \times \R \times \R^+$. The elements of $\mathcal{C}$ are called \textit{cylinders}, since any $(\gamma_x, t, l) \in \mathcal{C}$ can be seen as a cylinder on $S \times \R$ of axis $\{x\}\times [t,t+l]$ and diameter $\gamma$. However, we shall prefer to describe each cylinder $C=(\gamma_x,t,s)\in \mathcal{C}$ in terms of its \textit{basis} $\gamma_x$, its \textit{time of birth} $t$ and its \textit{lifespan} $l$. We denote these three features of $C$ by $basis(C)$, $b_C$ and $l_C$, respectively. 

In the following we consider particle configurations that belong either to $\C$ or to $\C \times [0,1]$. This requires the obvious adaptation of all the definitions in Section \ref{sec:pc}. We perform the necessary adjustments in the following definition.

\begin{defi}\label{def:ext} Let $(X,d)$ be a locally compact complete separable metric space. 
\begin{enumerate}
	\item [$\bullet$] A measure $\theta$ on $(X,\B_{X})$ is called a Radon measure if $\theta(B) < +\infty$ for any $B \in \B^0_{X}$. It is called a particle configuration on $X$ if in fact $\theta(B) \in \N_0$ for any $B \in \B^0_{{X}}$.
	\item [$\bullet$] The space of particle configurations on $X$ is denoted by $\mathcal{N}^*(X)$.\footnote{Notice that in Section \ref{sec:pc} we asked the configurations in $\mathcal{N}(S \times G)$ to be of locally finite allocation. Hence the extra index $*$ in $\mathcal{N}^*(X)$. This restriction was imposed in Section \ref{sec:pc} for simplicity of notation in the statement of Theorem \ref{convdis}, but it is unnecessary here.}
	\item [$\bullet$] $\mathcal{N}^*(X)$ is endowed with a measurable space structure by considering the $\sigma$-algebra generated by the counting events on $X$. 
	\item [$\bullet$] Given a Radon measure $\vartheta$ on $X$ we define the Poisson distribution $\pi^\vartheta$ as the unique probability measure on $\mathcal{N}^*(X)$
	which satisfies
	$$
	\pi^{\vartheta} ( \{ \theta \in \mathcal{N}^*(X) : \theta(B_i) = k_i \text{ for all }i=1,\dots,n \} ) = \prod_{i=1}^n \frac{e^{-\vartheta(B_i)} \left(\vartheta(B_i)\right)^{k_i}}{k_i!}
	$$ for all $k_1,\dots,k_n \in \N_0$, disjoint $B_1,\dots,B_n \in \B^0_{X}$ and $n \in \N$.
	\item [$\bullet$] A random particle configuration on $X$ is called a Poisson process  with intensity $\vartheta$ if it is distributed according to $\pi^\vartheta$.
\end{enumerate} 
\end{defi}
In the following, we will often identify a given random cylinder configuration $\mathcal{V}$ with a birth-and-death process on $S \times G$ through its time sections: if for each $t \in \R$ we define the random particle configuration $\mathcal{V}_t \in \mathcal{N}(S \times G)$ by the formula
$$
\mathcal{V}_t ( \{ \gamma_x\})  := \#\{ C \in \mathcal{V} : basis(C)=\gamma_x \text{ and }b_C \leq t < b_C + l_C \}
$$ for every $\gamma_x \in S \times G$, then $\mathcal{V}=(\mathcal{V}_t)_{t \in \R}$ constitutes a birth-and-death process on $S \times G$. From this point of view, we interpret any cylinder $(\gamma_x,t,l)$ as an animal $\gamma$ born at time $t$ on location $x$ with a lifetime of length $l$. 

Now, consider a Poisson process $\Pi$ on $\mathcal{C}$ with intensity measure $\phi_\nu: = \nu \times e^{-\Delta E} \mathcal{L} \times \mathcal{E}^1$, where $\mathcal{L}$ is the Lebesgue measure on $\R$ and $\mathcal{E}^1$ is the exponential distribution \mbox{of mean 1.} We call $\Pi$ the \textit{free process}, since it is a non-interacting birth-and-death process on $S \times G$. It is stationary and has $\pi^{e^{-\Delta E}\nu}$ as its invariant measure. The process $\mathcal{K}$ will be obtained as an appropriate thinning of $\Pi$. However, to properly conduct such thinning we need to add an additional component to $\Pi$: to each cylinder in $\Pi$ we attach an independent uniform random variable, which we call its \textit{flag}. Each of these flags will be used to determine the success of its cylinder's attempted birth in the dynamics. More formally, we define the \textit{flagged free process} $\overline{\Pi}$ as the Poisson process on $\mathcal{C} \times [0,1]$ with intensity measure $\overline{\phi}_\nu :=\phi_\nu \times \mathcal{L}_{[0,1]}$. For any $(\gamma_x,t,l) \in \Pi$ we denote its corresponding flag by $F(\gamma_x,t,l)$. Thus, elements of $\overline{\Pi}$ are simply pairs of the form $(C,F(C))$ with $C \in \mathcal{C}$. Finally, we define the thinned process $\mathcal{K}$ by the formula
\begin{equation}\label{keptformula2}
\mathcal{K} = \{ (\gamma_x,t,l) \in \Pi : F(\gamma_x,t,l) \leq M(\gamma_x | \mathcal{K}_{t^-}) \}
\end{equation} where, for $\gamma_x \in S \times G$ and $\xi \in \mathcal{N}(S\times G)$ we use the notation $M(\gamma_x|\xi):=e^{-(\Delta E_{\xi}(\gamma_x)-\Delta E)}$. Observe that the self-referential nature of the thinning rule in \eqref{keptformula2} could keep the process $\mathcal{K}$ from being well-defined. Indeed, let us introduce some definitions that will help us give further details on this matter.

\begin{defi}\label{defiances}$\,$
	\begin{enumerate}
		\item [$\bullet$] Given $C, \tilde{C} \in \C$ we say that $\tilde{C}$ is a \textit{first generation ancestor} of $C$ and write $\tilde{C} \rightharpoonup C$ whenever
		$$
		basis(\tilde{C}) \rightharpoonup basis(C)\hspace{1.5cm} \text{ and }\hspace{1.5cm} b_{\tilde{C}} < b_C < b_{\tilde{C}} + l_{\tilde{C}},
		$$ 
		where $\rightharpoonup$ is the impact relation defined in \eqref{eq:rob1}. We will denote by $\mathcal{P}(C)$ \mbox{the set of} all first generation ancestors of a given $C \in \C$.
		\item [$\bullet$] For $C \in \C$ we define $\A_1(C):=\Pi_{\mathcal{P}(C)}$, the restriction of $\Pi$ to $\mathcal{P}(C)$, and for $n \in \N$ we set
		$$
		\A_{n+1}(C):= \bigcup_{\tilde{C} \in \mathcal{A}_n(C)} \A_1(\tilde{C}).
		$$ We define the \textit{clan of ancestors} of $C$ in $\Pi$ as
		$$
		\A(C):= \bigcup_{n \in \N} \A_n(C).
		$$ 
		\item [$\bullet$] For $t \in \R$ and $\Lambda \in \B^0_S$ let us define the \textit{clan of ancestors of $\Lambda \times G$ at time $t$} as
		$$
		\mathcal{A}^t(\Lambda \times G) := \bigcup_{n \in \N_0} \mathcal{A}^t_{n}(\Lambda \times G)
		$$ where $\mathcal{A}^t_0(\Lambda \times G) := \{ C \in \Pi : basis(C) \in \Lambda \times G \text{ , } b_C \leq t < b_C + l_C \}$ and for $n \in \N$
		$$
		\mathcal{A}^t_n(\Lambda \times G) := \bigcup_{C \in \mathcal{A}^t_0 (\Lambda \times G)} \mathcal{A}_n(C).
		$$ 
	\end{enumerate}
\end{defi}

Let us return to the discussion of the well-definiteness of the process $\mathcal{K}$. Notice that if we wish to determine whether a given cylinder $C=(\gamma_x,t,l) \in \Pi$ belongs to $\mathcal{K}$ or not then first we need to specify the configuration $\mathcal{K}_{t^-}$ in order to evaluate whether the condition in \eqref{keptformula2} is satisfied. To be more accurate, due to Assumptions \ref{assump} we will only need to specify $\mathcal{K}_{t^-}$ inside the set $I(\gamma_x)$. Hence, recalling Definition \ref{defiances}, we see that \mbox{to determine} the fate of $C$ we must first determine the fate of, in principle, all of its ancestors in $\A_1(C)$. \mbox{But this itself} involves determining the fate of a second generation of ancestors, $\A_2(C)$. In general, to determine if $C$ belongs to $\mathcal{K}$ one may need to study the fate of every cylinder in the clan of ancestors of $C$. If $\A(C)$ were to span over an infinite number of generations then it may be impossible to decide whether to keep $C$ or not and, hence, the process $\mathcal{K}$ would not be well-defined in this situation. On the other hand, if for every cylinder $C \in \Pi$ the clan $\A(C)$ spans only over a finite number of generations then $\mathcal{K}$ would effectively be well-defined. Indeed, since $M(\gamma_x | \mathcal{K}^\sigma_{t^-}) = M(\gamma_x | \emptyset)$ for any $(\gamma_x, t,l) \in \Pi$ with no ancestors preceding it, the fate of every cylinder in the last generation of ancestors of a given $C \in \Pi$ can be decided upon inspecting their respective flags (and nothing else). By proceeding one generation at a time, the fate of all their descendants, including $C$, can be determined. More precisely, take $C \in \Pi$ and define
$$
N_C= \max\{ n \in N: \mathcal{A}_n(C) \neq \emptyset\}.
$$ If $N_C < +\infty$ set
$$
K_{N_C}(C):= \{ (\tilde{\gamma}_y,r,l) \in \mathcal{A}_{N_C}(C) : F(\tilde{\gamma}_y,r,l)< M(\tilde{\gamma}_y|\emptyset)\}
$$ and for $1 \leq i \leq N_C-1$ inductively define
$$
K_i(C) = K_{i+1}(C) \cup \{ (\tilde{\gamma}_y,r,l) \in \mathcal{A}_i(C) : F(\tilde{\gamma}_y,r,l)< M(\tilde{\gamma}_y|K_{i+1}(C))\}.
$$ Then the cylinder $C \in \Pi$ is kept in $\mathcal{K}$ if and only if
$$
F(C) < M(\tilde{\gamma}_y|K_1(C)).
$$ This algorithm decides whether $C$ is kept in finitely many steps. Therefore, we have that for $\mathcal{K}$ to be well-defined it suffices to have $N_C < +\infty$ for every $C \in \Pi$. As it turns out, the heavy diluteness condition of Definition \ref{defi:heavy} is enough to ensure this.

\begin{prop}[{\cite[Proposition~9.3]{SJS}}] \label{finit1} If $(\nu,H)$ is heavily diluted then almost surely the clans of ancestors $\mathcal{A}^t(\Lambda \times G)$ are finite for every $t \in \R$ and $\Lambda \in \B^0_S$.
\end{prop}

It follows from the discussion above and a straightforward computation that if $(\nu,H)$ is heavily diluted then $\mathcal{K}$ is well-defined and its infinitesimal generator is given by \eqref{generator}. Furthermore, Proposition \ref{finit1} implies the following facts:\begin{itemize}
	\item[i.] The dynamics defining $\mathcal{K}$ loses memory of the initial condition, a fact which implies uniqueness of the invariant measure. Furthermore, since $\mathcal{K}$ is a stationary process (due to its time-translational invariant construction and the stationarity of $\Pi$), this implies that for each $t \in \R$ the configuration $\mathcal{K}_t$ is distributed according to the unique invariant measure of the dynamics. 
	\item[ii.] $\mathcal{K}^{\Lambda|\eta}_t \overset{loc}{\longrightarrow}  \mathcal{K}_t$ as $\Lambda \nearrow S$ for all $t \in \R$ and $\eta \in \mathcal{N}(S \times G)$ of finite $H$-interaction range, where $\mathcal{K}^\Lambda$ is the process obtained by running the dynamics in the volume $\Lambda$ with boundary condition $\eta$ (and with $\Delta E_\xi$ replaced by $\Delta E_{\Lambda|\xi_{\Lambda}\cdot\eta_{\Lambda^c}}$ in the definition of $M$, see \ref{keptformula2}). Each process $\mathcal{K}^{\Lambda|\eta}$ is again stationary and a straightforward calculation with its generator shows that $\mu_{\Lambda|\eta}$ is its unique invariant measure. From this and the local convergence stated above, we conclude that the invariant measure of $\mathcal{K}$ is the unique gas measure of $(\nu,H)$.
\end{itemize}


\section{Proofs of Theorem \ref{convdis} and Corollary \ref{convabs}}\label{secdis}

In this final section we give the proof of our main result, Theorem \ref{convdis}. Following the proof, \mbox{we discuss} some relaxations of its hypotheses and then conclude the section by showing how Corollary \ref{convabs} follows from Theorem \ref{convdis}.

\subsection{Proof of Theorem \ref{convdis}} We divide the proof in three steps.

\subsubsection{Uniqueness of the gas measure for each model $(\nu^\varepsilon,H^\varepsilon)$}

Let us begin by considering a Poisson process $\overline{\Pi}$ on $\mathcal{C} \times [0,1]$ with intensity measure $\overline{\phi}_\nu:=\nu \times e^{-\Delta E} \mathcal{L} \times \mathcal{L}_{\R^+} \times \mathcal{U}[0,1]$. Notice that, by the nature of the intensity, all cylinders in $\overline{\Pi}$ have multiplicity one. Thus, it makes sense to define for each $\varepsilon \geq 0$ the $\varepsilon$-discretized process $\overline{\Pi}^\varepsilon$ (or simply $\varepsilon$-process) by the formula
\begin{equation}\label{piepsilon}
\overline{\Pi}^\varepsilon := \{ ( \gamma_x^\varepsilon, t, s, u ) \in \mathcal{C} \times [0,1] : (\gamma_x, t, s, u ) \in \overline{\Pi}\}.
\end{equation} Observe that $\overline{\Pi}^\varepsilon$ is a Poisson process on $\mathcal{C} \times [0,1]$ with intensity $\nu^\varepsilon \times e^{-\Delta E} \mathcal{L} \times \mathcal{L}_{\R^+} \times \mathcal{U}[0,1]$.
Moreover, \eqref{piepsilon} establishes a one-to-one correspondence between cylinders of \mbox{$\Pi$ and $\Pi^\varepsilon$.} Thus, in the following we write $C_\varepsilon$ to denote the $\varepsilon$-cylinder in $\Pi^\varepsilon$ which corresponds to the cylinder $C \in \Pi$, i.e.
if $C=(\gamma_x,t,s)$ then we set $C_\varepsilon = (\gamma_x^\varepsilon,t,s )$.

Now, by the proof of Theorem \ref{teounigibbs} we see that it suffices to show that for each $\varepsilon \geq 0$ $\mathcal{A}^{0,H^\varepsilon}(\Lambda \times G)$, the clan of ancestors at time 0 with respect to the Hamiltonian prescription $H^\varepsilon$ and with underlying free process $\Pi^\varepsilon$, is finite almost surely for all $\Lambda \in \B^0_S$. To see this, notice that if we define an impact relation $\rightharpoonup_V$ by the condition 
$$
\tilde{\gamma}_y \rightharpoonup_V \gamma_x \Longleftrightarrow \tilde{\gamma_y} \in V(\gamma_x)
$$ and for $\Lambda \in \B^0_S$ we consider the corresponding clan of ancestors $\mathcal{A}^{0,V}(\Lambda \times G)$ \mbox{given by $\rightharpoonup_V$,} then for each $\varepsilon \geq 0$ we have 
\begin{equation}\label{ancesincdis}
\mathcal{A}^{0,H^\varepsilon}(\Lambda \times G) \subseteq D_\varepsilon \left( \A^{0,V}(\Lambda_{\varepsilon} \times G)\right)
\end{equation} where $\Lambda_{\varepsilon}=\{ x \in S : d(x,\Lambda) \leq \varepsilon\}$ and for $\Gamma \subseteq \mathcal{C}$ we set
$$
D_\varepsilon (\Gamma) = \{ (\gamma^\varepsilon_x , t,s ) \in \mathcal{C} : (\gamma_x,t,s) \in \Gamma \}.
$$ By the proof of Proposition \ref{finit1} (see \cite{SJS} for details) we see that $\alpha^{\nu^0,V}_q < 1$ implies that almost surely $\A^{0,V}(\Lambda \times G)$ is finite for every $\Lambda \in \B^0_S$, so that from \eqref{ancesincdis} we conclude that $\mathcal{A}^{0,H^\varepsilon}(\Lambda \times G)$ must be finite as well. This shows that $(\nu^\varepsilon,H^\varepsilon)$ has a unique gas measure. Notice that our argument does not imply that the model $(\nu^\varepsilon,H^\varepsilon)$ is heavily diluted itself.  Indeed, the only available estimate  
$$
\frac{e^{-\Delta E^{H^\varepsilon}}}{q(\gamma_x^\varepsilon)} \int_{I^{H^\varepsilon}(\gamma_x^\varepsilon)} q(\tilde{\gamma}_y) d\nu^\varepsilon(\tilde{\gamma}_y) \leq \frac{e^{-\Delta E}}{q(\gamma_x^\varepsilon)} \int_{V(\gamma_x)} q(\tilde{\gamma}_y^\varepsilon) d\nu^0(\tilde{\gamma}_y)
$$ for every $\gamma_x \in S\times G$ is, in principle, not enough to show that $\alpha^{\nu^0,V}_q < 1$ implies $\alpha^{\nu^\varepsilon,H^\varepsilon}_q < 1$.

\subsubsection{Construction of the coupling $(\mathcal{Z}^\varepsilon)_{\varepsilon \geq 0}$}

We need first the following technical lemma.

\begin{lema}\label{lema:nim} Suppose that $N \subseteq \mathcal{N}(S \times G)$ is a $\sigma$-local $\pi^\nu$-null event. Then 
	\begin{equation}
	\label{din}
	P( \Pi_t \in N \text{ for some $t \in \R$} )=0.
	\end{equation}
\end{lema}

\begin{proof} If $N= \bigcup_{n \in \N} N_n$ for a sequence $(N_n)_{n \in \N}$ of local events, then it suffices to show that for each $n \in \N$ one has
	$$
	P( \Pi_t \in N_n \text{ for some $t \in \R$} )=0.
	$$ Thus, let us fix $n \in \N$ and consider $\Lambda \in \B^0_S$ such that $N_n \in \F_{\Lambda \times G}$. Let us observe that since $N_n$ is $\F_{\Lambda \times G}$-measurable we have 
	\begin{equation}\label{eq:nim2}
	\Pi_t \in N_n \Longleftrightarrow \left(\Pi_t\right)_{\Lambda \times G} \in N_n.
	\end{equation} Now, let $M_\Lambda$ be the set of cylinder configurations $\theta \in \mathcal{N}^*(\Lambda \times G \times \R \times \R^+)$ such that:
	\begin{enumerate}
		\item [i.] $\theta_t(\Lambda \times G) <+\infty$ for all $t \in \R$.
		\item [ii.] Every cylinder in $\theta$ has a strictly positive lifespan, i.e $l_C > 0$ for all $C \in \theta$. 
	\end{enumerate} Notice that if $\theta \in M_\Lambda$ then the process $\theta= (\theta_t)_{t \in \R}$ is piecewise constant. In particular, for $\theta \in M_\Lambda$ we have 
	$$
	\theta_t \in N_n \text{ for some $t \in \R$} \Longleftrightarrow \theta_r \in N_n \text{ for some $r \in \mathbb{Q}$}.
	$$ 
	Now, the choice of the intensity measure $\phi_\nu$ plus the fact that $\nu(\Lambda \times G) < +\infty$ yields that 
	$$
	P( \Pi_{\Lambda \times G \times \R \times \R^+}  \in M_\Lambda)=1.
	$$ 
	Hence, from \eqref{eq:nim2} we obtain that
	\begin{align*}
	P( \Pi_t \in N_n \text{ for some $t \in \R$} )&= P( \left(\Pi_r\right)_{\Lambda \times G} \in N_n \text{ for some $r \in \mathbb{Q}$} )\\
	\\
	& \leq \sum_{r \in \mathbb{Q}} P( \Pi_r \in N_n) = \sum_{r \in \mathbb{Q}} \pi^\nu(N_n) = 0.
	\end{align*}
\end{proof}

Now, we construct the coupling $(\mathcal{Z}^\varepsilon)_{\varepsilon \geq 0}$ by considering the processes
$$
\mathcal{K}^\varepsilon = \{ (\gamma_x^\varepsilon,t,s) \in \Pi : F(\gamma_x^\varepsilon,t,s) < M^\varepsilon(\gamma_x^\varepsilon | \mathcal{K}^\varepsilon_{t^-}) \}
$$ where for each $\gamma_x \in S \times G$ and $\xi \in \mathcal{N}(S \times G)$ we define
$$
M^\varepsilon (\gamma_x |\xi) := e^{- (\Delta E^{H^\varepsilon}_\xi(\gamma_x) - \Delta E)}.
$$ By the arguments in Section \ref{secffg}, each process $\mathcal{K}^\varepsilon$ is stationary with invariant measure $\mu^\varepsilon$. Thus, for each $\varepsilon \geq 0$ we may define $\mathcal{Z}^\varepsilon := \mathcal{K}^\varepsilon_0$. We need to check that for any $B \in \B^0_{S \times G}$ there exists $\varepsilon_B > 0$ such that for every $\varepsilon \leq \varepsilon_B$ 
\begin{equation}
\label{eq:igualdad}\mathcal{Z}^\varepsilon_B = D_\varepsilon\left( \mathcal{Z}^0_{D^{-1}_\varepsilon(B)}\right).
\end{equation}
To see this, we notice that if $N$ denotes the dynamically negligible event in the statement of Theorem \ref{convdis} then by Lemma \ref{lema:nim} and since $\mathcal{K}^0$ is a thinning of $\Pi$ we can assume that, with the exception of a set $O$ of realizations of $\Pi$ with zero probability, for every $C \in \Pi$ we have 
\begin{equation} \label{eq:nim3}
\left(\mathcal{K}^0_{b_C}\right)_{I(basis(C))} + \delta_{basis(C)} \in N^c.
\end{equation} What we will show in fact is that, for almost every realization of the process $\Pi$ outside of this exceptional set $O$, for any given $B \in \B^0_{S \times G}$ if $\varepsilon$ is sufficiently small then \eqref{eq:igualdad} holds. The heart of the proof is contained in the next lemma, which states that clans of ancestors stabilize for $\varepsilon$ small. 

\begin{lema}\label{lema:heart}
Fix $\Lambda \in \B^0_S$ and let $O$ be the exceptional set above. Then in $O^c$ there exists almost surely (a random) $\varepsilon_0 > 0$ such that for $0 \leq \varepsilon < \varepsilon_0$
\begin{equation}\label{ancesincdisc1}
\mathcal{K}^\varepsilon_{D_\varepsilon(\A^{0,V} ( \Lambda \times G))} = D_\varepsilon \left(\mathcal{K}^0_{\A^{0,V} ( \Lambda \times G)}\right),
\end{equation} 
i.e. any $C \in \mathcal{A}^{0,V}(\Lambda \times G)$ is kept in $\mathcal{K}^0$ if and only if $C_\varepsilon$ is kept in $\mathcal{K}^\varepsilon$ for every $0 \leq \varepsilon < \varepsilon_0$.
\end{lema}

\begin{proof} Define $K:=\max\left\{ n \in \N : \A^{0,V}_n(\Lambda \times G) \neq \emptyset\right\}$, which is finite since $\alpha_q^{\nu^0,V} < 1$. \mbox{Notice that} for every cylinder $C \in \A^{0,V}_K(\Lambda \times G)$ and $\varepsilon \geq 0$ one has 
$$
C_\varepsilon \in \mathcal{K}^\varepsilon \Longleftrightarrow F(C) < M^\varepsilon( basis(C_\varepsilon) | \emptyset ).
$$	Thus, since $\emptyset + \delta_{basis(C)} \in N^c$ by \eqref{eq:nim3} and $I(basis(C))=\emptyset$ by definition of $K$, by recalling that
$$
\lim_{\varepsilon \rightarrow 0^+} \Delta E^{H^\varepsilon}_{\eta^\varepsilon} (\gamma^\varepsilon_x) = \Delta E^{H^0}_\eta (\gamma_x)
$$ 
for every $\gamma_x \in S \times G$ and $\eta \in \mathcal{N}(S \times G)$ with $\eta + \delta_{\gamma_x} \in N^c$, we deduce that for $\varepsilon$ (randomly) small enough  
$$
\mathcal{K}^\varepsilon_{D_\varepsilon(\A^{0,V}_N ( \Lambda \times G))} = D_\varepsilon \left(\mathcal{K}^0_{\A^{0,V}_N( \Lambda \times G)}\right),
$$ except perhaps if $F(C) = M^0(basis(C)|\emptyset)$ for some cylinder $C \in \A^{0,V}_N(\Lambda \times G)$, \mbox{a fact which} can only occur with zero probability. 
In a similar way one may proceed with the following generations by induction, using at each step \eqref{eq:nim3} and \eqref{ancesincdis}, to arrive ultimately at \eqref{ancesincdisc1}.
\end{proof}

Now, fix $B \in \B^0_{S \times G}$ and take $\Lambda \in \B^0_S$ sufficiently large so that $B^{(1)} \subseteq \Lambda \times G$, where $B^{(1)}$ is as in Definition \ref{discdefi}. By Lemma \ref{lema:heart} applied to this volume $\Lambda$ and the inclusion 
$$
D^{-1}_\varepsilon (K) \subseteq \Lambda \times G
$$ 
valid for every $0 <\varepsilon \leq 1$, we conclude that \eqref{eq:igualdad} holds for $\varepsilon$ small enough. 

\subsubsection{Concluding the weak convergence $\mu^\varepsilon \to \mu^0$}
The weak limit now follows from \eqref{eq:igualdad} by the dominated convergence theorem and the following lemma.

\begin{lema}\label{lemaconvaga} Let us suppose that $(\xi^{(\varepsilon)})_{\varepsilon \geq 0} \subseteq \mathcal{N}(S \times G)$ satisfies that for any $B \in \B^0_{S \times G}$ there exists $\varepsilon_B > 0$ such that $\xi^{(\varepsilon)}_B = D_\varepsilon ( \xi^{(0)}_{D^{-1}_\varepsilon(B)})$ for all $\varepsilon \leq \varepsilon_B$. Then $\xi^{(\varepsilon)} \rightarrow \xi^{(0)}$ vaguely.
\end{lema}

\begin{proof} It suffices to show that for each compact set $K \subseteq S \times G$ and $\delta > 0$ there exists $\varepsilon_0 > 0$ small enough such that $\xi^{(\varepsilon)} \in (\xi^{(0)})_{K,\delta}$ for all $0 <\varepsilon < \varepsilon_0$. But if given $B \in \B^0_{S \times G}$ we define
	$$
	\rho_B:= \frac{1}{2}\min\{ d_{S \times G}\left(\gamma_x,\tilde{\gamma}_y\right) : \gamma_x\neq \tilde{\gamma}_y \in \langle \xi^{(0)}_{B} \rangle\}
	$$ and take
	$$
	\varepsilon_0 := \rho_{K^{(\delta)}} \wedge \varepsilon_{K_{(\delta)}} \wedge \delta > 0
	$$ where $K^{(\delta)}$ and $K_{(\delta)}$ are the sets from Definition \ref{discdefi}, then $\xi^{(\varepsilon)} \in (\xi^{(0)})_{K,\delta}$ for every $\varepsilon < \varepsilon_0$, since:
	\begin{enumerate}
		\item [i.] From the definition $\varepsilon_{K_{(\delta)}}$ and the fact that $K \subseteq K_{(\delta)}$ we see that every $\gamma^{(\varepsilon)}_x \in \langle \xi^{(\varepsilon)}_K\rangle$ is of the form $\gamma^{(\varepsilon)}_x = D_\varepsilon(\gamma_x)$ for some $\gamma_x \in \langle \xi^{(0)} \rangle$. Moreover, since $\varepsilon < \rho_{K^{(\delta)}}$ \mbox{there is} at most one such $\gamma_x$ so that, in particular, their multiplicities must be the same. Thus, the application $p: [\xi^{(\varepsilon)}_K] \to [\xi]$ given by 
		$$
		p(\gamma_x^{(\varepsilon)},i)=\left(D_\varepsilon^{-1}(\gamma_x^{(\varepsilon)}),i\right)
		$$ is well-defined and injective. Since $\varepsilon < \delta$ we obtain that $\xi^{(\varepsilon)}_K \preceq_\delta \xi^{(0)}$.
		\item [ii.] By definition of $\varepsilon_{K_{(\delta)}}$ and the fact that $K \subseteq D^{-1}_\varepsilon(K_{(\delta)})$ we see that for every $\gamma_x \in \langle \xi^{(0)}_K \rangle$ there exists $\gamma^{(\varepsilon)}_x \in \langle \xi^{(\varepsilon)}\rangle$ such that $D_\varepsilon(\gamma_x) = \gamma^{(\varepsilon)}_x$. Furthermore, since $ \varepsilon < \rho_{K^{(\delta)}}$ there is at most one $\gamma_x \in \langle\xi^{(0)}_K\rangle$ being mapped by $D_\varepsilon$ to $\gamma^{(\varepsilon)}_x$ so that, in particular, the multiplicity of $\gamma^{(\varepsilon)}_x$ in $\xi^{(\varepsilon)}$ is at least equal to that of $\gamma_x$ in $\xi^{(0)}_K$. Therefore, the application $p:[\xi^{(0)}_K] \to [\xi^{(\varepsilon)}]$ given by
		$$
		p(\gamma_x, i)=\left(D_\varepsilon(\gamma_x),i\right)
		$$ is well-defined and injective. Since $\varepsilon < \delta$ this shows that $\xi^{(0)}_K \preceq_\delta \xi^{(\varepsilon)}$.
	\end{enumerate}
\end{proof} 

\subsection{Some relaxations in the hypotheses of Theorem \ref{convdis}} Notice that for the proof of Theorem \ref{convdis} we only required the dynamically negligible event $N$ to satisfy the following two properties:

\begin{enumerate}
	\item [$\bullet$] $N$ is closed under the addition of particles.
	\item [$\bullet$] $P( \Pi_t \in N \text{ for some $t \in \R$} )=0.$
\end{enumerate} Thus, if we call any event $N$ verifying these properties a \textit{dynamically impossible} event, then condition (ii) in the hypotheses of Theorem \ref{convdis} may be relaxed by requiring that the event in question be only dynamically impossible. However, we should point out that \mbox{in most cases} of interest the original condition is already satisfied and it is simple to verify, so that there is no real gain from this relaxation. Furthermore, we also point out that in the statement of the theorem all infima and suprema may be taken to be essential in the measure-theoretical sense, arriving at the more general statement of Theorem \ref{convdis} below. 

\begin{teo}
Let $(D_\varepsilon)_{\varepsilon \geq 0}$ be an approximation family and suppose that $(\nu^\varepsilon,H^\varepsilon)_{\varepsilon \geq 0}$ is a family of diluted models such that:
\begin{enumerate}
\item[I.]For every $\varepsilon > 0$ the intensity measure $\nu^\varepsilon$ is given by 
	$$
	\nu^\varepsilon := \nu^0 \circ D_\varepsilon^{-1}.
	$$
	\item [II.] There exist a $\nu^0$-null set $M \subseteq S \times G$ and a measurable set $N \subseteq \mathcal{N}(S \times G)$ closed under the addition of particles satisfying:
	\begin{enumerate}
		\item [i.]
		$$
		\Delta E := \inf_{\varepsilon \geq 0} \left[ \inf_{\Lambda \in \B^0_S} \left[ \inf_{\substack{ \gamma_x \in M^c \\ \eta + \delta_{\gamma_x} \in N^c }}\Delta E^{H^\varepsilon}_{\eta^\varepsilon}(\gamma_x^\varepsilon)\right] \right] > -\infty.
		$$ 
		\item[ii.] $P( \Pi_t \in N \text{ for some $t \in \R$})=0$, where $\Pi$ is the Poisson process on $\mathcal{C}$ with intensity measure $\nu^0 \times e^{-\Delta E} \mathcal{L} \times \mathcal{E}^1$, where $\Delta E$ is as in (i). [Notice that $\Delta E$ (and hence $\Pi$) depends, in principle, on the particular choice of $N$.]  
		\item [iii.] For every $\gamma_x \in M^c$ and all $\eta \in \mathcal{N}(S \times G)$ with $\eta + \delta_{\gamma_x} \in N^c$ one has
		$$
		\lim_{\varepsilon \rightarrow 0^+} \Delta E^{H^\varepsilon}_{\eta^\varepsilon} (\gamma^\varepsilon_x) = \Delta E^{H^0}_\eta (\gamma_x).
		$$ 	
		\item [iv.] There exists for each $\gamma_x \in S \times G$ a set $V(\gamma_x) \in \B_{S \times G}$ such that:
		\begin{itemize}
			\item[$\bullet$] For every $\varepsilon \geq 0$ one has the inclusion
			$$
			D_\varepsilon^{-1}\left(I^{H^\varepsilon}(\gamma^\varepsilon_x)\right) \subseteq V(\gamma_x).
			$$
			\item[$\bullet$] There exists a size function $q: S \times G \rightarrow [1,+\infty)$ which verifies 
			$$
			\alpha_{q}^{\nu^0,V} := \sup_{\gamma_x \in S \times G} \left[ \frac{e^{-\Delta E}}{q(\gamma_x)} \int_{V(\gamma_x)} q(\tilde{\gamma}_y)d\nu^0(\tilde{\gamma}_y)\right] < 1.
			$$
		\end{itemize}
\end{enumerate}
\end{enumerate}
Then the conclusions of Theorem \ref{convdis} hold, namely: 
\begin{itemize}
	\item[(a)] Each model $(\nu^\varepsilon,H^\varepsilon)$ admits exactly one gas measure $\mu^\varepsilon$.
	\item[(b)]  As $\varepsilon \rightarrow 0^+$, we have the weak convergence
	$$
	\mu^{\varepsilon} \overset{w}{\longrightarrow} \mu^0.
	$$ 
	\item[(c)] There exists a coupling $(\mathcal{Z}^\varepsilon)_{\varepsilon \geq 0}$ of the measures $(\mu^\varepsilon)_{\varepsilon \geq 0}$ such that for any $B \in \B^0_{S \times G}$ there exists (a random) $\varepsilon_B > 0$ verifying that for all $\varepsilon \leq \varepsilon_B$
	\begin{equation}
	\label{eq:couple}
	\mathcal{Z}^\varepsilon_B = D_{\varepsilon}\left( \mathcal{Z}^0_{D^{-1}_\varepsilon(B)}\right).
	\end{equation} 
	In particular, $\mathcal{Z}^\varepsilon \overset{as}{\longrightarrow} \mathcal{Z}^0$ with respect to the vague topology.
\end{itemize}
\end{teo}

\subsection{Proof of Corollary \ref{convabs}}
 
To conclude this last section, we show that Corollary \ref{convabs} is a particular case of Theorem \ref{convdis}. Indeed, notice that, since $\nu^\varepsilon \ll \nu$, for every $\Lambda \in \B^0_S$ we have that $\pi^{\nu^\varepsilon}_\Lambda \ll \pi^{\nu}_\Lambda$ with density given by
\begin{equation}\label{density}
\frac{d\pi^{\nu^\varepsilon}_\Lambda}{d\pi^{\nu}_\Lambda} (\sigma) = e^{-( \nu^\varepsilon(\Lambda \times G) - \nu(\Lambda \times G))} \prod_{(\gamma_x,i) \in [\sigma]} \frac{d\nu^\varepsilon}{d \nu}(\gamma_x),
\end{equation} a fact which can be deduced from the integration formula (see {\cite[Proposition~3.1]{MW}})
$$
\int_{\mathcal{N}(\Lambda \times G)} f(\sigma) d\pi^{\upsilon}_\Lambda = \sum_{n \in \N_0} \frac{e^{-\upsilon(\Lambda \times G)}}{n!} \int_{(\Lambda \times G)^n} f\left( \sum_{i=1}^n \delta_{\gamma_x^i}\right) d\upsilon^n(\gamma_x^1,\dots,\gamma_x^n)
$$ valid for all bounded $\F_{\Lambda\times G}$-measurable functions $f$ and measures $\upsilon$ on $S \times G$ of $S$-locally finite allocation, where $\upsilon^n$ above denotes the $n$-fold product measure of $\upsilon$. In particular, the models $(\nu,H^\varepsilon)$ and $(\nu,\tilde{H}^\varepsilon)$ are equivalent, i.e. they produce the same gas kernels $\mu^\varepsilon_{\Lambda |\eta}$, where $\tilde{H}^\varepsilon$ is the Hamiltonian prescription given by 
$$
\tilde{H}^\varepsilon_{\Lambda|\eta} (\sigma) := H^\varepsilon_{\Lambda|\eta}(\sigma) - \sum_{(\gamma_x,i) \in [\sigma]} \log\left( \frac{d \nu^\varepsilon}{d \nu}(\gamma_x)\right).
$$ It is not difficult to check that for each $\varepsilon \geq 0$ the pair $(\nu,\tilde{H}^\varepsilon)$ satisfies Definition \ref{assump}. Moreover, since for every $\gamma_x \in S \times G$ it is possible to verify that $I^{\tilde{H}^\varepsilon}(\{\gamma_x\}) = I^{H^\varepsilon}(\{\gamma_x\})$ and we also have that 
$$
\Delta \tilde{E}^{\varepsilon}_{\Lambda|\eta} = \Delta E^{\tilde{H}^\varepsilon}_{\Lambda|\eta}
$$ for every $\Lambda \in \B^0_S$ and $\eta \in \mathcal{N}(S\times G)$, we see that the family of models $(\nu,\tilde{H}^\varepsilon)_{\varepsilon \geq 0}$ falls under the hypotheses of Theorem \ref{convdis} for the approximation family $\mathcal{D}=(D_\varepsilon)_{\varepsilon \geq 0}$ given by the identity operator on $S \times G$ for each $\varepsilon \geq 0$. The result then follows from Theorem \ref{convdis}.  

\bigskip
\noindent {\bf Acknowledgments.} The authors would like to thank Pablo Ferrari and Siamak Taati for the many fruitful discussions and suggestions that helped us make this a better article.

\bibliographystyle{plain}
\bibliography{bibliografia}
\end{document}